\documentclass[lettersize,journal]{IEEEtran}
\usepackage[utf8]{inputenc}
\usepackage{indentfirst}
\usepackage{amsfonts}
\usepackage{amssymb}
\usepackage{algorithm}
\usepackage{algorithmic}
\usepackage{graphicx}
\usepackage{float} 
\usepackage{subfigure}
\usepackage{caption}
\usepackage{cite}
\usepackage{verbatim}
\usepackage{diagbox}
\usepackage{threeparttable}
\usepackage{multirow}
\usepackage{stfloats}
\usepackage{makecell}
\usepackage{xcolor}
\usepackage{amsmath}
\usepackage{amsthm}
\newtheorem{theorem}{Theorem}

\newtheorem{proposition}[theorem]{\indent Proposition}
\newtheorem{lemma}{Lemma}

\ifCLASSINFOpdf

\fi

\begin{document}

    \title{RDARS Empowered Massive MIMO System: Two-Timescale Transceiver Design with Imperfect CSI}
\author{Chengzhi Ma,
        Jintao Wang,
        Xi Yang,
        Guanghua Yang, 
        Wei Zhang,
        and Shaodan Ma
\thanks{C. Ma, J. Wang and S. Ma are with the State Key Laboratory of Internet of Things for Smart City and the Department
of Electrical and Computer Engineering, University of Macau, Macao SAR, China (e-mails: yc07499@um.edu.mo; 
wang.jintao@connect.um.edu.mo;
shaodanma@um.edu.mo).}
\thanks{X. Yang is with the Shanghai Key Laboratory of Multidimensional Information Processing, School of Communication and Electronic Engineering, East China Normal University, Shanghai 200241, China (e-mail: xyang@cee.ecnu.edu.cn).}
\thanks{G. Yang is with the School of Intelligent Systems Science and Engineering and GBA and B{\&}R International Joint Research Center for Smart Logistics, Jinan University, Zhuhai 519070, China (e-mail: ghyang@jnu.edu.cn).}
\thanks{W. Zhang is with the School of Electrical Engineering \& Telecommunications, University of New South Wales, Sydney, Australia (e-mail: w.zhang@unsw.edu.au).}
}
\UseRawInputEncoding


\maketitle
\begin{abstract} 


    In this paper, we investigate a novel reconfigurable distributed antennas and reflecting surface (RDARS) aided multi-user massive multiple-input multiple-output (MIMO) system with imperfect channel state information (CSI) and propose a practical two-timescale (TTS) transceiver design to reduce the communication overhead and computational complexity of the system. In the RDARS-aided system, not only distribution gain but also reflection gain can be obtained by a flexible combination of the distributed antennas and reflecting surface, which differentiates the system from the others and also makes the TTS design challenging. To enable the optimal TTS transceiver design, the achievable rate of the system is first derived in closed-form. The rate expression is general and covers that of the distributed antenna systems (DAS) and reconfigurable intelligent surface (RIS) aided systems as special cases. Then the TTS design aiming at the weighted sum rate maximization is considered. To solve the challenging non-convex optimization problem with high order design variables, i.e., the transmit powers and the phase shifts at the RDARS, a \textit{block coordinate descent} based method is proposed to find the optimal solutions in semi-closed forms iteratively. Specifically, two efficient algorithms are proposed with provable convergence for the optimal phase shift design, i.e., \textit{Riemannian Gradient Ascent} based algorithm by exploiting the unit-modulus constraints, and \textit{Two-Tier Majorization-Minimization} based algorithm with closed-form optimal solutions in each iteration. Simulation results validate the effectiveness of the proposed algorithm and demonstrate the superiority of deploying RDARS in massive MIMO systems to provide substantial rate improvement with a significantly reduced total number of active antennas/RF chains and lower transmit power when compared to the DAS and RIS-aided systems. 

\end{abstract}


\IEEEpeerreviewmaketitle

\section{Introduction}

   Massive multiple-input multiple-out (MIMO), by utilizing a large number of antennas at the transceiver, particularly at the base station, to exploit rich spatial diversity and multiplexing gains for spectrum efficiency enhancement, has become a key enabling technique to be adopted in the future sixth-generation (6G) wireless communications. Combining the massive MIMO with state-of-the-art radio-over-fiber (ROF) technologies \cite{Lisu_Yu}, the massive active antennas can be deployed in a centralized or distributed manner to support diverse communication needs in various applications. When deployed in distributed manner, the system is also termed to be a distributed antenna system (DAS) \cite{Lin_Dai}. It can achieve additional distribution gain due to the richer scattering environments and thus is anticipated to play an essential role in future wireless communication systems.  
   
 
   To achieve the desirable high spectrum efficiency to meet the ever-increasing communication needs with stringent quality-of-service requirements, a large number of antennas and radio-frequency chains are generally required to be deployed \cite{Lisu_Yu, Xi_Yang}. This undoubtedly introduces tremendous hardware cost and energy consumption, thus limiting the practical applications of massive MIMO. On the other hand, reconfigurable intelligent surface (RIS), also known as intelligent reflecting surface (IRS), has been proposed as a promising technology to improve spectrum efficiency (SE) and energy efficiency (EE) of the communication network while maintaining low cost and energy consumption advantages\cite{Qingqing_Wu, Chongwen_Huang, Buzzi}. Specifically, RIS is a planar array composed of many passive elements, each of which can impose an independent phase shift on the incident signals. By smartly tuning the phase shifts, the wireless propagation environment becomes controllable and programmable, which thus provides an extra degree of freedom to optimize communication performance. However, the performance gain of directly applying RIS to current wireless communication systems is limited by the ``multiplicative fading'' effect and generally only exhibited when the direct link is relatively weak or suffers from blockage, which erodes the practical benefits \cite{Cunhua_Pan_Overview}. Moreover, the gain is obtained only when the phase shifts are well aligned with the propagation channels. With a large number of passive reflecting elements, the RIS control overhead is inevitably high, particularly under time varying channel environments,  which also limits its practical applications. 
    

    
    Thankfully, embedded with the merits of both DAS and RIS, reconfigurable distributed antenna and reflecting surface (RDARS) has been proposed recently as a promising new architecture for wireless communications \cite{RDARS}. 
    Specifically, a RDARS is a planar array consisting of reconfigurable elements, where each element can switch between two modes, namely, the \textit{connected mode} and the \textit{reflection mode}. When working under the \textit{connected mode}, the elements act as distributed antennas connected to the BS via dedicated wires or fibers and can receive the incoming or transmit wireless signal as in the DAS \cite{Lin_Dai}. When working under \textit{reflection mode}, the elements work as passive reflecting elements where each of them can induce a distinct phase shift to the incident signal, thus collaboratively changing the reflected signal propagation as an element in RIS \cite{Chongwen_Huang, Qingqing_Wu, Jintao_Wang}. Consequently, RDARS encompasses DAS and RIS as special cases, offering flexible trade-off between \textit{distribution gain} and \textit{reflection gain} to enhance the performance, with great potential for practical deployment. Prior work has demonstrated the superiority of RDARS-aided systems through theoretical analysis and experimental results, and initially verified its potential to greatly reduce the number of active antennas/RF chains and power consumption with a substantial gain under the single-user scenario \cite{RDARS, RDARS_demo}. 
    


    Nonetheless, the potential of RDARS under multi-user scenario is not well discovered yet. Its analysis is more challenging than that for single user scenario since multi-user interference is involved. Besides, to fully unleash the potential and achieve substantial performance gain of the RDARS-aided multi-user systems, the joint transceiver design of the beamforming for combining the signals at the BS and the distributed antennas at the RDARS, users' transmit powers, and the phase shifts of the reflecting elements on RDARS is necessary and yet to be investigated. Similarly to the RIS-aided systems, channel state information (CSI) is generally required for transceiver optimization so as to achieve the promising gain in RDARS-aided systems. However, due to the involvement of a large number of elements in RDARS, the number of channel parameters to be estimated scales with the product of the number of antennas at the BS and the number of elements in the RDARS. This would generally lead to a prohibitively high pilot overhead \cite{Cunhua_Pan_Overview}. Moreover, phase shift design based on instantaneous CSI (I-CSI) requires frequent calculation and feedback for RDARS configuration in each channel coherence interval. These generally result in high computational complexity, feedback overhead, and power consumption \cite{Kangda_Zhi}.  

    To address the aforementioned practical challenges and to reap the promising gain, a two-timescale (TTS) transceiver design mechanism which was first proposed for RIS-aided systems \cite{MingMin_Zhao} is exploited in this paper. In such a TTS design, the beamformer at BS is configured based on the I-CSI of \textit{equivalent channel}, i.e., the user-BS end-to-end channel, which includes direct and reflected links. Accordingly, the dimension of the channel matrix to be estimated is almost the same as in the RDARS-free system, which is independent of the number of RDARS elements. Hence, the number of pilot signals required only scales linearly with the number of users, significantly reducing the pilot overhead. Meanwhile, such TTS mechanism aims to optimize the phase shifts based on long-term statistical CSI (S-CSI). As such, the phase shifts on the RDARS are reconfigured when the large-scale channel information changes, which generally varies much slower than I-CSI for typical applications, especially in the sub-6 GHz band \cite{Cunhua_Pan_Overview}. Feedback overhead for RDARS configuration is thus expected to be significantly reduced. The promising overhead reduction and performance gain of the TTS design have been unveiled theoretically for the RIS-aided system \cite{Kangda_Zhi}. However, due to the different architecture of RDARS from RIS, the performance gain and the optimal TTS design for RDARS-aided systems are yet to determined. Notice that even for RIS-aided systems, optimal phase shifts in the TTS design has not been found for multi-user uplink scenarios due to the challenges in high-order parameter optimization. Moreover, due to the limited training pilots, channel estimation errors are unavoidable in practice. The impact of imperfect CSI on the RDARS-aided multi-user systems with TTS transceiver design is meaningful to be discovered for practical applications. These motivate the work in this paper, which analyzes the RDARS-aided multi-user systems and investigates the optimal TTS design with imperfect CSI. 

    \vspace{-1.3cm}
    \subsection{Main Contributions}
    The main contributions of this paper can be summarized as follows:
        \begin{itemize}
            \item We propose to invoke RDARS to empower uplink transmission in multi-user massive MIMO systems. To fully discover the potential of the RDARS-aided multi-user systems under practical scenarios, we first propose a practical channel estimation scheme based on the linear minimum mean square error (LMMSE) criteria to obtain the uplink \textit{equivalent channel} state information. The channel estimator only requires a limited number of pilots, which is independent of the number of elements on RDARS, thus is appealing for practical applications. We then analyze the achievable rate of the RDARS-aided system with a maximum ratio combiner (MRC) based on the estimated CSI. The derived rate expression holds for any finite numbers of antennas at the BS and the elements on RDARS. It also generalizes the achievable rate expressions derived for the RIS-aided  MIMO system \cite{Kangda_Zhi}, the conventional colocated MIMO system \cite{SIG-093}, as well as the DAS system with two distributed antenna sets \cite{DAS}. The closed-form expression of the achievable rate not only provides an effective analytical approach to evaluate the performance and understand the behaviors of the RDARS-aided multi-user systems, but also enables transceiver optimization to further exploit the potential of the RDARS-aided systems for performance enhancement.
            
           \item Based on the derived closed-form achievable rate, a 
            weighted sum rate maximization for the joint design of users' transmit power and phase shifts on RDARS is formulated with a TTS design to reduce the overhead and complexity in the long term. The optimization problem is challenging to solve since the the objective function involves the \textit{``fourth-order"} term w.r.t. the phase shifts. Besides, the coupling of variables further complicates the problem. As a result, such a problem has not been considered in the previous works.\footnote{To the best of our knowledge, even for the special case of RDARS-aided system, i.e., RIS-aided system, the weighted sum rate optimization with joint design of phase shifts and users' transmit powers considering both imperfect CSI and MRC detector under TTS design framework is not available yet. Only phase optimization is conducted in \cite{Kangda_Zhi} to maximize the minimum user rate. As such, our proposed optimization algorithm also serves as another building block for studying RIS-aided MIMO systems.}

            \item An effective solution is proposed to tackle the challenging optimization problem. Specifically, we first reformulate the original problem into an equivalent one using the \textit{Fractional Programming}(\textit{FP}) technique. Then the non-convex \textit{block coordinate descent} (\textit{BCD}) method is invoked to decouple the variables and enable the iterative optimization of each variable. In each block, closed-form optimal solutions for the users' transmit powers and the dual variables are derived. With respect to the optimization of the phase shifts, two effective algorithms are proposed. One is a \textit{Riemannian Gradient Ascent} (RGA) algorithm. It exploits the unit-module constraints imposed on phase shifts and finds the optimal phase shifts numerically through line searching. The other is a \textit{Majorization-Minimization} (\textit{MM})-based method. It gives the closed-form solution for the optimal phase shifts after transforming the original \textit{``fourth-order"} phase shift optimization problem into a tractable form by further leveraging the inherent structure of the objective function using sophisticated matrix manipulations and transformations.  Computational complexity analysis is provided with provable convergence to stationary solutions. Moreover, the proposed algorithms are applicable to both the RIS-aided and DAS systems.
            
            \item We verify the effectiveness of the proposed design through simulations under various scenarios. The results show that significant performance gain can be attained through optimization with the proposed algorithms. Moreover, the results demonstrate the superiority of applying RDARS for multi-user massive MIMO systems by significantly reducing the users' transmit powers and the number of active antennas/RF chains with substantial gain when compared to its counterparts. 
        \end{itemize}
        
    The rest of the paper is organized as follows. Section II gives the system model and the TTS framework. Then, the LMMSE channel estimation is introduced in Section III. In Section IV, the performance of the underlying system with imperfect CSI is analyzed and a closed-form achievable rate is derived. In Section V, the weighted sum rate maximization is discussed and semi-closed-form solutions for the transmit powers and phase shifts are derived based on S-CSI only. Simulation results are given in Section VI, while conclusions are drawn in Section VII.
    
    \textit{Notations:} Throughout the paper, scalars, vectors, and matrices are denoted by lower-case, bold-face lower-case, and bold-face upper-case letters, respectively. $(\cdot)^{T}$, $(\cdot)^{H}$ denote the transpose and conjugate transpose of a matrix or vector. ${\rm{diag}}(\mathbf{v})$ denotes a diagonal matrix with each diagonal element being the corresponding element in $\mathbf{v}$. Similarly, ${\rm{blk}}(\cdot)$ denotes the block-diagonalization operation. Furthermore, $|\cdot|$ denotes the modulus of the input complex number. ${\rm{Tr}}(\cdot)$, $\mathbb{E}[\cdot]$ and $\mathbb{C}[\cdot, \cdot]$ represent trace, expectation and covariance operator, respectively. The element in the $i^{th}$ row and $j^{th}$ column of matrix $\mathbf{A}$ is denoted as $\mathbf{A}(i,j)$. $\lambda_{max}(\cdot)$ retrieves the maximum eigenvalue of the matrix, and $\mathbf{I}_a$ denotes the identity matrix with dimension of $a \times a$. 

\vspace{-6pt}
\section{system Model}

We consider a RDARS-aided massive MIMO multi-user communication system. As shown in Fig. 1, the system consists of one BS with $L$ antennas, $K$ single-antenna users, and one RDARS with $N$ elements. For simplicity, the user index set is $\mathcal{K} = \{1, 2, \cdots, K \}$. An indicator matrix $\mathbf{A} \in \mathbb{R}^{a \times N}$ is introduced to specify the modes of elements on RDARS, where $a \leq N$ is the total number of elements performing \textit{connected mode} on RDARS. Specifically, $\mathbf{A}^{H}\mathbf{A}$ is a diagonal matrix, and $\mathbf{A}^{H}\mathbf{A}(i,i) = 1$ means the $i^{th}$ element works at \textit{connected mode} on RDARS and is programmed to connect with BS through high-quality fronthaul, while $\mathbf{A}^{H}\mathbf{A}(j,j) = 0$ indicates that the $j^{th}$ element operates at \textit{reflection mode}, i.e., the $j^{th}$ element can only alter the phase of incoming signals. For clarity, we define $\mathcal{A} = \{ i | \mathbf{A}^{H}\mathbf{A}(i,i) = 1 \}$ and $\widetilde{\mathcal{A}} = \{ i | \mathbf{A}^{H}\mathbf{A}(i,i) = 0 \}$ as the index sets for elements performing \textit{connected mode} and \textit{reflection mode}, respectively. Let $\mathbf{\Theta} = {\rm{diag}}(\boldsymbol{\theta}) \in \mathbb{C}^{N \times N}$ denote the RDARS reflection-coefficient matrix whose diagonal element is expressed as $\mathbf{\Theta}(i, i) = \boldsymbol{\theta}(i) =e^{j{\rm{arg}}(\boldsymbol{\theta}(i))}, \forall i \in \widetilde{A}$, where $\boldsymbol{\theta}(i)$ denotes the $i^{th}$ element of $\boldsymbol{\theta}$ and ${\rm{arg}}(\boldsymbol{\theta}(i))$ is the phase shift the $i^{th}$ element will induce. 

With the help of RDARS, the received signal at the BS and the elements acting \textit{connected mode} on RDARS, $\mathbf{y} \in \mathbb{C}^{(L+a) \times 1}$, can be compactly written as
    \vspace{2pt}
    \begin{align} \label{y}
    \begin{footnotesize}
    \begin{aligned}
        \setlength{\abovedisplayskip}{20pt}
        \mathbf{y} = \sum\nolimits_{k=1}^{K} \sqrt{p}_k \mathbf{q}_k x_{k} + \widetilde{\mathbf{n}} = \mathbf{Q} \mathbf{x} + \widetilde{\mathbf{n}},
        \setlength{\belowdisplayskip}{20pt}
    \end{aligned}
    \end{footnotesize}
    \end{align}
\noindent where $\mathbf{Q} = [\mathbf{q}_1, \mathbf{q}_2, \cdots, \mathbf{q}_K] \in \mathbb{C}^{(L+a) \times K}$ is the aggregated \textit{equivalent channel} matrix. $\mathbf{q}_k$ is the \textit{equivalent channel} corresponding to the $k^{th}$ user and we have $\mathbf{q}_k \triangleq [\mathbf{h}_{B,k}^{T} \ \mathbf{h}_{R,k}^{T}]^{T} = [(\underline{\mathbf{h}}_{B,k} + \mathbf{d}_{k})^{T} \ (\mathbf{A}\mathbf{h}_k)^{T}]^{T} \in \mathbb{C}^{ (L+a) \times 1}$, where $\mathbf{h}_{B,k} \triangleq \underline{\mathbf{h}}_{B,K} +  \mathbf{d}_{k} \triangleq  \mathbf{H} \mathbf{B} \mathbf{h}_{k} +  \mathbf{d}_{k}$, $\mathbf{h}_{R,k} \triangleq \mathbf{A}\mathbf{h}_k$, and $\mathbf{B} \triangleq (\mathbf{I}-\mathbf{A}^{H}\mathbf{A})\mathbf{\Theta}$. Here, $\mathbf{d}_k \in \mathbb{C}^{L \times 1}$, $\mathbf{h}_k \in \mathbb{C}^{N \times 1}$ and $\mathbf{H} \in \mathbb{C}^{L \times N}$ are the channels from the $k^{th}$ user to the BS, the $k^{th}$ user to the RDARS, and the RDARS to the BS, respectively. $p_k$ denotes the transmit power of the $k^{th}$ user and $p_k \leq p_{\max}, \forall k \in \mathcal{K}$, where $p_{\max}$ denotes the maximum transmit power. $x_k$ denotes the transmit symbol of the $k^{th}$ user and we have
$\mathbf{x} = [\sqrt{p}_1 x_1, \sqrt{p}_2 x_2, \cdots, \sqrt{p}_K x_K]^T $. $\widetilde{\mathbf{n}} \triangleq [\mathbf{n}_B \ \mathbf{n}_R]^{T} \in \mathbb{C}^{(L+a)\times 1}$ represents the noise vector where $\mathbf{n}_B \sim \mathcal{CN}(\mathbf{0},\sigma_B^2\mathbf{I}_{L})$ and $\mathbf{n}_R \sim \mathcal{CN}(\mathbf{0},\sigma_R^2\mathbf{I}_{a})$ are the additive white Gaussian noises (AWGNs) received at the BS antennas and the elements performing \textit{connected mode} on RDARS, respectively. 

Clearly, the BS receives two sets of signals, one from the BS antennas directly and the other from RDARS elements operating at the \textit{connected mode}. To obtain the distribution gain, the BS then applies maximum ratio combining based on the estimate of the equivalent channel matrix  (i.e., $\hat{\mathbf{Q}}$) to recovery the user signals. The equivalent channel can be obtained by the LMMSE estimator which will be introduced later. Thus, based on \eqref{y}, the recovered signal is given as $\mathbf{r} =  \hat{\mathbf{Q}}^{H} \mathbf{Q}\mathbf{x} + \hat{\mathbf{Q}}^{H}  \widetilde{\mathbf{n}}$. Accordingly, for the $k^{th}$ user, the $k^{th}$ element of $\mathbf{r}$ can be expressed as
    \begin{align} \label{r_k 1}
    \begin{footnotesize}
    \begin{aligned}
        r_k = \sqrt{p_k} \hat{\mathbf{q}}_k^{H} \mathbf{q}_k x_{k} 
        + 
        \sum_{\substack{i=1, i \neq k}}^{K} \sqrt{p_i} \hat{\mathbf{q}}_k^{H} \mathbf{q}_i x_{i} 
        +
        \hat{\mathbf{q}}_k^{H}  \widetilde{\mathbf{n}}, \forall k \in \mathcal{K}, 
    \end{aligned}
    \end{footnotesize}
    \end{align}
where $\hat{\mathbf{q}}_k \in \mathbb{C}^{(L+a)\times 1}$ is the $k^{th}$ column of $\hat{\mathbf{Q}}$. The equivalent channel $\mathbf{q}_k$ is different from that of the DAS and RIS-aided systems, but covers them as special cases. Specifically, it has a more complex structure since multiple links are involved, i.e., the direct link from the user to the BS antennas, the distributed link from the user to the RDARS elements operating at the \textit{connected mode}, and the reflected link from the user to the BS antennas via the RDARS reflecting elements. The presence of multiple links is expected to offer not only distribution gain but also reflection gain to boost the system performance. However, it also complicates the analysis and transceiver design of the RDARS-aided systems. In the following, we first proceed with the channel estimation to obtain an estimation of $\mathbf{q}_k$ and then conduct the performance analysis to shed light on the potential of such RDARS-aided massive MIMO systems.

\begin{figure} [t]  
        \small
        \setlength{\belowcaptionskip}{-3mm}  
        \centering
        \includegraphics[width=0.80\columnwidth]{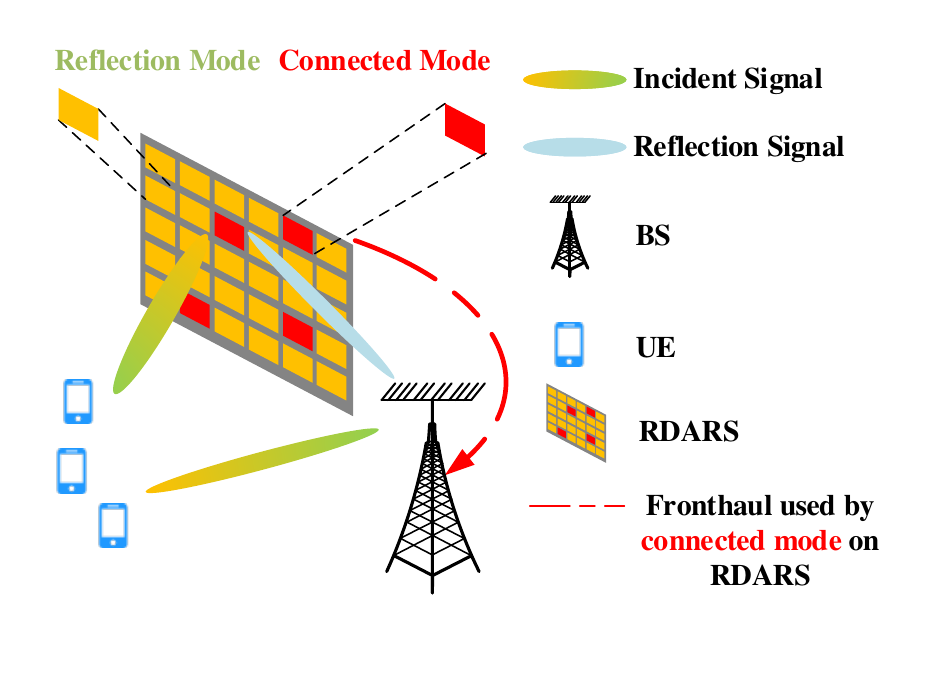} 
        \caption{ 
        A RDARS-aided uplink multi-user massive MIMO communication system.}     
        \label{Fig.1}
\end{figure}

\vspace{-7pt}
\subsection{Channel Model}
    Considering that the RDARS could be deployed to create line-of-sight (LoS) links between the user and RDARS or RDARS and BS to assist the communication when the direct user-BS link is blocked, we adopt the Rician fading model for the user-RDARS and RDARS-BS links while the Rayleigh fading model is leveraged for the direct user-BS link as in \cite{Kangda_Zhi}. The RDARS-BS, user-RDARS, and user-BS channel models are expressed as 
    $\mathbf{H} 
    =
        \sqrt{\frac{\beta}{\delta+1}} (\sqrt{\delta} \overline{\mathbf{H}} + \widetilde{\mathbf{H}})$,
    $\mathbf{h}_k 
    = 
        \sqrt{\frac{\alpha_k}{\epsilon_k+1}} (\sqrt{\epsilon_k} \overline{\mathbf{h}}_k + \widetilde{\mathbf{h}}_k)$,
    $\mathbf{d}_k 
    =
       \sqrt{\gamma_k} \widetilde{\mathbf{d}}_k$, respectively,   
    where $\delta$ and $\epsilon_k$ are the Rician factors of the RDARS-BS and user-RDARS links, respectively. $\beta$, $\alpha_k$, and $\gamma_k$ are the path-loss coefficients for RDARS-BS, user-RDARS, and user-BS links, respectively, and $\widetilde{\mathbf{H}} \in \mathbb{C}^{L \times N}$, $\widetilde{\mathbf{h}}_k \in \mathbb{C}^{N \times 1}$, and $\widetilde{\mathbf{d}}_k \in \mathbb{C}^{L \times 1}$ represent the NLoS components of the RDARS-BS, user-RDARS and user-BS links whose components are i.i.d. complex Gaussian random variables with zero mean and unit variance. $\overline{\mathbf{H}} \in \mathbb{C}^{L \times N}$ and $\overline{\mathbf{h}}_k \in \mathbb{C}^{N \times 1}$ denote the LoS components of the RDARS-BS and user-RDARS links, respectively. The structures of $\overline{\mathbf{H}}$ and $\overline{\mathbf{h}}_k$ depend on the array geometry. In this paper, we adopt uniform planar array (UPA) for both BS and RDARS with the size of $L = L_x \times L_y$ and $N = N_x \times N_y$, respectively. Therefore, we have
    $\overline{\mathbf{H}} = \mathbf{a}_L(\phi_{RB}^{a}, \phi_{RB}^{e}) 
            \mathbf{a}_N^{H}(\varphi_{RB}^{a}, \varphi_{RB}^{e})$ and
    $\overline{\mathbf{h}}_k = \mathbf{a}_N(\varphi_{UR,k}^{a}, \varphi_{UR,k}^{e})$, where $\varphi_{UR,k}^{a}$ ($\varphi_{UR,k}^{e}$) is the azimuth (elevation) angle of the arrival (AoA) of the incident signal at the RDARS from $k^{th}$ user, $\varphi_{RB}^{a}$ ($\varphi_{RB}^{e}$) is the azimuth (elevation) angle of the departure (AoD) reflected by the RDARS towards BS, and $\phi_{RB}^{a}$ ($\phi_{RB}^{e}$) is the azimuth (elevation) AoA of the signal received at BS from the RDARS, respectively. Besides, $\mathbf{a}_{X}(\psi^a, \psi^e) \in \mathbb{C}^{X \times 1}$ with $ X \in \{ L, N \}$ and $\psi \in \{ \phi_{RB}, \varphi_{RB}, \varphi_{UR}  \}$ denotes the array response vector, whose $x$-entry is
    \vspace{-3pt}
    \begin{align}
    \begin{footnotesize}
    \begin{aligned}
        [\mathbf{a}_{X}(\psi^a, \psi^e)](x) 
        &= {\rm{exp}} \bigg\{ 
            j2\pi\frac{d_s}{\lambda}
        \bigg(
                \lfloor (x-1)/X_y \rfloor \sin \psi^a \sin \psi^e
        \\
        &+
                ((x-1){\rm{mod}}X_y) \cos \psi^e
            \bigg)
        \bigg\},
    \end{aligned}
    \end{footnotesize}
    \end{align}
    \noindent where $d_s$ and $\lambda$ denote the element spacing of the antenna array and the wavelength, respectively.

\vspace{-7pt}
\subsection{Two-timescale Transceiver Design}

An illustration of TTS transceiver design framework is presented in Fig. 2. Specifically, the RDARS configuration and the uplink transmit powers are designed based on the S-CSI in the long term. In contrast, the MRC performed at the BS is implemented based on the I-CSI in the short term. As a result, the TTS transceiver design can significantly reduce the computational complexity, feedback overhead, and energy consumption, by alleviating frequent phase shift configuration and information feedback required in each channel coherence interval \cite{Cunhua_Pan_Overview}. 


\vspace{-10pt}
\begin{figure} [!htb]  
        \setlength{\belowcaptionskip}{-3mm}   
        \includegraphics[width=1\columnwidth]{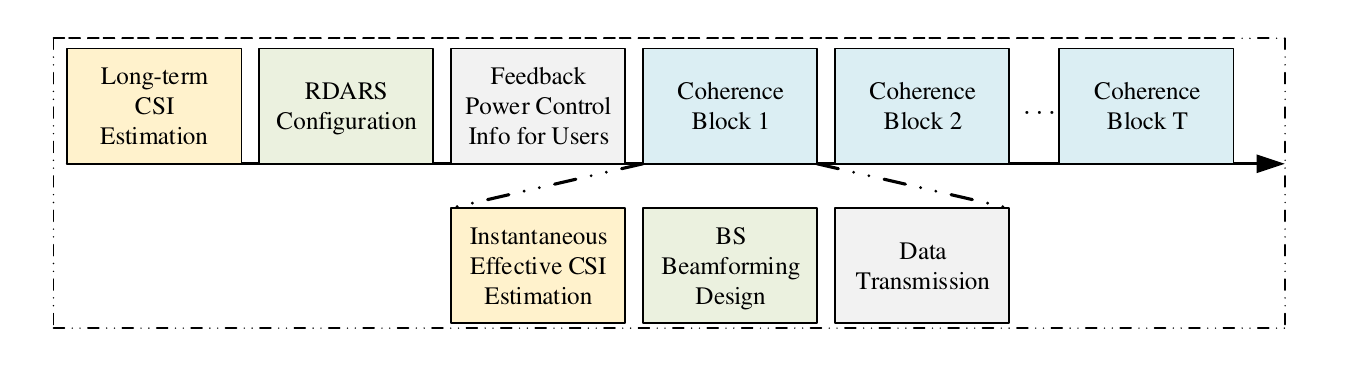} 
        \caption{ 
        The TTS transceiver design framework.}     
        \label{Fig.2}
\end{figure}
\vspace{-5pt}

\section{Channel Estimation}

In this section, the LMMSE channel estimator is introduced to estimate the \textit{equivalent channel} $\mathbf{Q}$. The orthogonal pilot sequence is leveraged for each user and the pilot sequence $\mathbf{s}_{k,p} \in \mathbb{C}^{\tau \times 1}, \forall k \in \mathcal{K},$ spans $\tau$ symbols in each coherence interval. It satisfies $\mathbf{S}^{H}\mathbf{S} = \mathbf{I}_K$ where $\mathbf{S} = [\mathbf{s}_{1,p}, \mathbf{s}_{2,p}, \cdots, \mathbf{s}_{K,p}] \in \mathbb{C}^{\tau \times K}$. Accordingly, the received pilot signal can be expressed as $\mathbf{Y}_p = [\mathbf{Y}_{p,B}^{T} \ \mathbf{Y}_{p,R}^{T}]^{T} =\sqrt{\tau p_p} \mathbf{Q} \mathbf{S}^{H} + \mathbf{N}$, where $\mathbf{Y}_{p,B} \in \mathbb{C}^{L \times \tau}$ and $\mathbf{Y}_{p,R} \in \mathbb{C}^{a \times \tau}$ are the received signals at the BS antennas and the RDARS elements acting on \textit{connected mode}, respectively. $p_p$ is the average transmit power of each pilot symbol. $\mathbf{N} = [\mathbf{N}_B \ \mathbf{N}_R]^T \in \mathbb{C}^{(L+a) \times \tau}$ is the noise matrix, and the entries of $\mathbf{N}_B$ and $\mathbf{N}_R$ are i.i.d. complex Gaussian random variables with zero mean and variances $\sigma_B^2$ and $\sigma_R^2$, respectively.

Based on the projection of $\mathbf{Y}_p$ onto $\frac{\mathbf{s}_{k,p}}{\sqrt{\tau p_p}}$, we have $\mathbf{y}_{k,p} \in \mathbb{C}^{(L+a) \times 1}$ as
    \begin{align} \label{y_k}
    \begin{footnotesize}
    \begin{aligned}
        \mathbf{y}_{k,p} = \frac{1}{\sqrt{\tau p_p}}\mathbf{Y}_p \mathbf{s}_{k,p} = \mathbf{q}_k + \frac{1}{\sqrt{\tau p_p}}\mathbf{N}\mathbf{s}_{k,p}.
    \end{aligned}
    \end{footnotesize}
    \end{align}
Then the LMMSE estimator can be derived from \eqref{y_k}. To proceed, the required statistics of the projected signals $\mathbf{y}_{k,p}$ and the equivalent channel is first provided as in the following lemma.


\begin{lemma} \label{Lemma1}
            For $k \in \mathcal{K}$, the mean vector and covariance matrices required to obtain the LMMSE channel estimator are given by
            \begin{align}
            \begin{footnotesize}
            \begin{aligned}
                \mathbb{E}[\mathbf{q}_k] = \mathbb{E}[\mathbf{y}_{k,p}] = 
                \begin{bmatrix}
                    \sqrt{c_k \delta \epsilon_k} \overline{\mathbf{H}}\mathbf{B} \overline{\mathbf{h}}_k \\
                    \sqrt{d_k \epsilon_k} \mathbf{A}\overline{\mathbf{h}}_k
                \end{bmatrix},
            \end{aligned}  
            \end{footnotesize}
            \end{align}

            \begin{align}
            \begin{footnotesize}
            \begin{aligned}
                \mathbb{C}[\mathbf{q}_{k}, \mathbf{y}_{k,p}] = 
                \begin{bmatrix}
                a_{k1} \mathbf{a}_L \mathbf{a}_L^{H} + a_{k2} \mathbf{I}_L
                     & \mathbf{0}_{L \times a} 
                     \\
                \mathbf{0}_{a \times L} & d_k\mathbf{I}_a
                \end{bmatrix},
            \end{aligned}  
            \end{footnotesize}
            \end{align}

            \begin{align}
            \begin{footnotesize}
            \begin{aligned}
                \mathbb{C}[\mathbf{y}_{k,p}, \mathbf{y}_{k,p}] = 
                \begin{bmatrix}
                a_{k1} \mathbf{a}_L \mathbf{a}_L^{H} + (a_{k2} + \frac{\sigma_B^2}{\tau p_p}) \mathbf{I}_L
                     & \mathbf{0}_{L \times a} 
                     \\
                \mathbf{0}_{a \times L} & (d_{k} + \frac{\sigma_{R}^2}{\tau p_p}) \mathbf{I}_a
                \end{bmatrix},
            \end{aligned}  
            \end{footnotesize}
            \end{align}
            where $a_{k1} \triangleq (N -a)c_k\delta$,$a_{k2} \triangleq ((N-a)c_k(\epsilon_k+1)+\gamma_k) $, $c_{k} \triangleq \frac{\beta\alpha_{k}}{(\delta+1)(\epsilon_{k}+1)}$ and $d_{k} \triangleq \frac{\alpha_{k}}{(\epsilon_{k}+1)}$.
    \end{lemma}

    \begin{proof}
        The proof can be found in Appendix A.
    \end{proof}


\vspace{-7pt}
Armed with the above statistics, the LMMSE estimate of the \textit{equivalent channel} $\mathbf{q}_k$ is derived in the following theorem.

    \begin{theorem} \label{Theorem 1}
        The LMMSE channel estimator of the \textit{equivalent channel} $\mathbf{q}_k$ is derived as
        \begin{align} \label{EstimatedChannel}
        \begin{footnotesize}
            \begin{aligned}
            \hat{\mathbf{q}}_{k} = [\hat{\mathbf{h}}_{B,k}^{T} 
            \ 
            \hat{\mathbf{h}}_{R,k}^{T}]^{T}, \forall k \in \mathcal{K},
        \end{aligned}  
        \end{footnotesize}
        \end{align}
        where $\hat{\mathbf{h}}_{B,k} = \hat{\underline{\mathbf{h}}}_{B,k} + \mathbf{E}_{k}  \mathbf{d}_{k} + \frac{1}{\sqrt{\tau p_p}} \mathbf{E}_{k}\mathbf{N}_{B} \mathbf{s}_{k}$, $\hat{\mathbf{h}}_{R,k} = \sqrt{d_k \epsilon_k} \mathbf{A} \overline{\mathbf{h}}_k + \sqrt{d_k}a_{k5}\mathbf{A} \widetilde{\mathbf{h}}_k + \frac{1}{\sqrt{\tau p_p}} a_{k5} \mathbf{N}_R \mathbf{s}_k$, and $\hat{\underline{\mathbf{h}}}_{B,k} = \sqrt{c_k \delta \epsilon_k} \overline{\mathbf{H}}\mathbf{B}\overline{\mathbf{h}}_k + (a_{k3}L + a_{k4}) \sqrt{c_k \delta} \overline{\mathbf{H}} \mathbf{B} \widetilde{\mathbf{h}}_k + \sqrt{c_k \epsilon_k} \mathbf{E}_k \widetilde{\mathbf{H}} \mathbf{B} \overline{\mathbf{h}}_k
        + \sqrt{c_k}\mathbf{E}_k \widetilde{\mathbf{H}} \mathbf{B} \widetilde{\mathbf{h}}_k$.
        Also, we have $a_{k3} = (a_{k1} \frac{\sigma_{B}^2}{\tau p_p} )/[(a_{k2} + \frac{\sigma_B^2}{\tau p_p}) ((a_{k2} + \frac{\sigma_{B}^2}{\tau p_p}  ) + L a_{k1} )]$,
        $a_{k4}=(a_{k2})/(a_{k2}+\frac{\sigma_{B}^2}{\tau p_p})$, $a_{k5} = d_k/(d_k + \frac{\sigma_R^2}{\tau p_p})$, and $\mathbf{E}_k = a_{k3}\mathbf{a}_L \mathbf{a}_L^{H} + a_{k4}\mathbf{I}_L$.
    \end{theorem}
    \begin{proof}
    See Appendix B.
    \end{proof}
    \vspace{-8pt}   
    Note that the equivalent channel matrix $\hat{\mathbf{Q}} \triangleq [\hat{\mathbf{q}}_1, \ldots, \hat{\mathbf{q}}_K]$ has a similar dimension as the channel matrix between the users and BS, i.e., the size of $\hat{\mathbf{Q}}$ is $(L+a) \times K$ ($a \ll N$). Thus, the required pilot symbols for obtaining $\hat{\mathbf{Q}}$ is almost the same as the conventional MIMO system without RDARS, i.e., $\tau \geq K$. Consequently, the training overhead for acquiring the I-CSI can be significantly reduced when compared to estimating the $LN$ individual channels \cite{Cunhua_Pan_Overview}.  

\vspace{-5pt}
\section{Analysis of Achievable Rate}

    Armed with $\hat{\mathbf{q}}_k$ provided in $\textbf{Theorem 1}$, we can now analyze the achievable rate of the RDARS-aided MIMO system in this section. The achievable rate is a widely adopted performance objective for system design. Its closed-form expression will enable transceiver optimization to further exploit the potential of the RDARS-aided systems for performance enhancement in Section V. 

\vspace{-10pt}
\subsection{Closed-Form Achievable Rate}

    Considering the equivalent channel $\mathbf{q}_k$ is not perfectly known, the received signal $r_k$ in (\ref{r_k 1}) can be rewritten as   
    \begin{align}
            \begin{footnotesize}
            \begin{aligned}
            r_{k} &= \underbrace{\sqrt{p_k} \mathbb{E}[\hat{\mathbf{q}}_k^{H} \mathbf{q}_k] x_k}_{ {\textit{Desired \ Signal}} }
            +
            \underbrace{\sqrt{p_k}(\hat{\mathbf{q}}_k^{H} \mathbf{q}_k - \mathbb{E}[\hat{\mathbf{q}}_k^{H} \mathbf{q}_k] ) x_{k} }_{ {\textit{Signal \ Leakage}} } 
            \\
            &+
            \underbrace{
            \sum\nolimits_{i\neq k}^{K} \sqrt{p_i} \hat{\mathbf{q}}_k^{H} \mathbf{q}_i x_i }_{ {\textit{Multi-user Interference}} }
            +
            \underbrace{
            \hat{\mathbf{q}}_k^{H} \widetilde{\mathbf{n}}}_{ {\textit{Noise}} }, \forall k \in \mathcal{K}.
            \end{aligned}
            \end{footnotesize}
    \end{align}
    \normalsize
    \noindent 
    Treating the second, the third, and the fourth terms as the \textit{effective noise}, it can be easily proved that the \textit{Desired \ Signal} is uncorrelated with the \textit{effective noise}.\footnote{Note that the phase shift configuration of RDARS configuration, i.e., $\boldsymbol{\theta}$, is fixed but unknown in this section, which will be designed by the proposed algorithms in the next section.}  As such, by applying the use-and-then-forget technique as in \cite{SIG-093, Kangda_Zhi}, the achievable rate of the $k^{th}$ user for the RDARS-aided system can be obtained as
    \begin{align} \label{Equation_Imperfect_CSI}
    \begin{footnotesize} 
    \begin{aligned}
        R_k = (\tau_c - \tau)/\tau \times  \log(1 + {\rm{SINR}}_k),
    \end{aligned}
    \end{footnotesize}
    \end{align}
    where $\tau_c$ denotes the length of the channel coherence interval, $(\tau_c - \tau)/\tau$ is the rate loss originating from the pilot overhead, and the signal-to-interference-noise-ratio (SINR) ${\rm{SINR}}_k$ can be given as \ref{EAR 1} as shown at the top of next page. To ease the analysis of the achievable rate in \ref{EAR 1}, we first characterize the auxiliary variable $\mathbf{E}_{k}$ introduced by the imperfect CSI in \textbf{Theorem 1} as follows.
    

    \newcounter{SINR}
        \begin{figure*}[!htb]
        \footnotesize
        \setcounter{SINR}{\value{equation}}
        \setcounter{equation}{10}
        \begin{align}
         {\rm{SINR}}_{k} = 
            \frac{
            p_k \big| \mathbb{E}[\hat{\mathbf{q}}_k^{H} \mathbf{q}_k] \big|^2
            }
            {
            p_k \bigg( \mathbb{E}\big[|\hat{\mathbf{q}}_k^{H} \mathbf{q}_k|^2\big] - \big|\mathbb{E}[\hat{\mathbf{q}}_k^{H} \mathbf{q}_k]\big|^2 \bigg)
            +
            \sum\nolimits_{i\neq k}^{K} p_i \mathbb{E}\big[|\hat{\mathbf{q}}_k^{H} \mathbf{q}_i|^2\big]
            +
            \mathbb{E}\big[ \hat{\mathbf{q}}_k^{H} {\rm{blk}}(\sigma_{B}^2 \mathbf{I}_{L}, \sigma_{R}^2 \mathbf{I}_{a} ) 
            \hat{\mathbf{q}}_k\big]
            } 
            \tag*{\normalsize (11)}
            \label{EAR 1}
        \end{align}
        \setcounter{equation}{\value{SINR}}
        \hrulefill
        \end{figure*}
        \addtocounter{equation}{1}

\begin{lemma} \label{Lemma2}
        We have ${\rm{Tr}}(\mathbf{E}_{k}) = L e_{k1}$, ${\rm{Tr}}(\mathbf{E}_{k}^{H} \mathbf{E}_{k}) = L e_{k3}$, ${\rm{Tr}}(\mathbf{a}_L^{H}\mathbf{E}_k \mathbf{E}_k^{H} \mathbf{a}_{L}) = L e_{k2}^2$, and $\mathbf{E}_k \overline{\mathbf{H}} = (La_{k3} + a_{k4})\overline{\mathbf{H}}=
            e_{k2} \overline{\mathbf{H}}$, $\forall k$, where
        $e_{k1} \triangleq a_{k3} + a_{k4}$, $e_{k2} \triangleq La_{k3} + a_{k4}$, $e_{k3} \triangleq L a_{k3}^2 + 2a_{k3}a_{k4} + a_{k4}^2$, and $e_{k4} \triangleq a_{k5}$.
        Besides, all the auxiliary variables are bounded within interval $[0,1]$. Notably, as $\tau p_p \rightarrow \infty$ , we obtain $e_{k1}, e_{k2}, e_{k3}, e_{k4} \rightarrow 1$. When $\tau p_p \rightarrow 0$, we obtain $e_{k1}, e_{k2}, e_{k3}, e_{k4} \rightarrow 0$.
    \end{lemma}
\begin{proof}
The proof can be derived by applying a similar technique as in \cite{Kangda_Zhi}, and thus omitted here for brevity.
\end{proof}

    \begin{theorem} \label{Theorem2}
        The $\rm{SINR}$ for the $k^{th}$ user ${\rm{SINR}}_k$ is then calculated as
            \begin{align} 
            \label{SINR_Theorem2}
            \begin{footnotesize}
            \begin{aligned}
                &{\rm{SINR}}_k 
                =
                                \frac{
                    p_k E_{k}^{signal}(\mathbf{A},\mathbf{\Theta})
                    }
                    {
                    p_k E_{k}^{leak}(\mathbf{A},\mathbf{\Theta})
                    +
                    \sum_{\substack{i=1 \\ i\neq k}}^{K} p_i I_{ki}(\mathbf{A},\mathbf{\Theta})
                    +
                    E_{k}^{noise}(\mathbf{A},\mathbf{\Theta})
                    },
            \end{aligned}
            \end{footnotesize}        
            \end{align}
            
    \noindent where $E_{k}^{signal}(\mathbf{A},\mathbf{\Theta})$, $E_{k}^{noise}(\mathbf{A},\mathbf{\Theta})$, $E_{k}^{leak}(\mathbf{A},\mathbf{\Theta})$ and $I_{ki}(\mathbf{A},\mathbf{\Theta})$ are given in \ref{Theorem2Signal}, \ref{Theorem2Noise}, \ref{Theorem2Leak} and \ref{Theorem2Interference} at the top of next page, respectively, with $ f_{k}(\mathbf{A},\mathbf{\Theta}) = \mathbf{a}_N^{H}\mathbf{B}\overline{\mathbf{h}}_k = \sum_{n=1}^{N} (1-a_n) e^{j(\zeta_n^k + \boldsymbol{\theta}(n))}$, $g_{k,i}(\mathbf{A}) = \overline{\mathbf{h}}_{k}^{H} \mathbf{A}^{H} \mathbf{A} \overline{\mathbf{h}}_{i}= \sum_{n=1}^{N} a_i e^{j \varsigma_{n}^{k,i}}$, $\zeta_n^k = 2\pi\frac{d}{\lambda}\big[ \lfloor (n-1)/\sqrt{N} \rfloor \big( {\rm{sin}}(\varphi_{UR,k}^e){\rm{sin}}(\varphi_{UR,k}^a) - {\rm{sin}}(\varphi_{RB}^e){\rm{sin}}(\varphi_{RB}^a) \big) + \big( (n-1){\rm{mod}}\sqrt{N} \big) \big({\rm{cos}}(\varphi_{UR,k}^e) - {\rm{cos}}(\varphi_{RB}^e) \big) 
    \big]$, and $\varsigma_{n}^{k,i} = 2\pi\frac{d}{\lambda}\big[ \lfloor (n-1)/\sqrt{N} \rfloor \big( {\rm{sin}}(\varphi_{UR,i}^e){\rm{sin}}(\varphi_{UR,i}^a) - {\rm{sin}}(\varphi_{UR,k}^e){\rm{sin}}(\varphi_{UR,k}^a) \big) + \big( (n-1){\rm{mod}}\sqrt{N} \big) \big({\rm{cos}}(\varphi_{UR,i}^e) - {\rm{cos}}(\varphi_{UR,k}^e) \big) \big]$.                
    \end{theorem}

    \begin{proof}
    See Appendix C.
    \end{proof}

    \newcounter{term}
        \begin{figure*}[!htb]
        \footnotesize
        \setcounter{term}{\value{equation}}
        \setcounter{equation}{12}
        \begin{align} 
            &E_{k}^{signal}(\mathbf{A},\mathbf{\Theta}) 
            = \bigg[ 
                L \big[
                    |f_{k}(\mathbf{A},\mathbf{\Theta}) |^2 c_{k} \delta \epsilon_{k}
                    +
                    (N-a) c_{k} \delta e_{k2}
                    +
                    \big( 
                        (N-a)c_{k}(\epsilon_{k} + 1) + \gamma_{k}
                    \big) e_{k1}
                \big]
            +
                a\big[ 
                     d_{k} \epsilon_{k}
                    +  e_{k4} d_{k}
                \big]
            \bigg]^2
        \tag*{\normalsize (13)}
        \label{Theorem2Signal}
        \\ 
            &E_{k}^{noise} (\mathbf{A},\mathbf{\Theta}) 
            =    
                   \sigma_B^2 L \big[
                        |f_{k}(\mathbf{A},\mathbf{\Theta}) |^2 c_{k} \delta \epsilon_{k}
                        +
                        (N-a) c_{k} \delta e_{k2}
                        +
                        \big( 
                            (N-a)c_{k}(\epsilon_{k} + 1) + \gamma_{k}
                        \big) e_{k1}
                    \big]
                +
                   \sigma_R^2 a \big[ 
                     d_{k} \epsilon_{k}
                    +  e_{k4} d_{k} 
                    \big]
        \tag*{\normalsize (14)}
        \label{Theorem2Noise}
        \\
        &E_{k}^{leak}(\mathbf{A},\mathbf{\Theta}) 
            = 
                L |f_{k}(\mathbf{A},\mathbf{\Theta})|^2
                    c_{k}^2 \delta \epsilon_k
                        \bigg[ 
                            (N-a) (L \delta + \epsilon_k + 1) 
                            (e_{k2}^2 + 1)
                            +
                            2(L e_{k1} + e_{k2})
                            (e_{k2} + 1)
                        \bigg]
            \notag
            \\
            &\ \ \ +
                L |f_{k}(\mathbf{A},\mathbf{\Theta})|^2
                c_k \delta \epsilon_k 
                \bigg[
                    \gamma_k 
                    + (\gamma_k + \frac{\sigma_B^2}{\tau \rho}) e_{k2}^2 
                \bigg]
            +
                L^{2} (N-a)^2 c_{k}^2 \delta^2 e_{k2}^2
            +
                2 L (N-a)^2 c_k^2 \delta (\epsilon_k + 1) e_{k2}^2
                +
                L (N-a)^2 c_k^2 (\epsilon_k + 1)^2 e_{k3}
            \notag
            \\
            &\ \ \ +
                L^2 (N-a) c_{k}^2 \bigg[ 
                    (2\epsilon_k + 1) e_{k1}^2
                    +
                    2 \delta e_{k1} e_{k2}
                \bigg] 
            +
                 L (N-a) c_{k} \bigg[ 
                    c_{k} \big( 
                        2\delta e_{k2}^2 + (2\epsilon_k + 1) e_{k3}
                    \big)
                    +
                    \big(
                        2 \gamma_k + \frac{\sigma_B^2}{\tau \rho}
                    \big)
                    (\delta e_{k2}^2 + (\epsilon_k + 1) e_{k3})
                \bigg]  
            \notag
            \\
            &\ \ \ +
                L \gamma_k 
                \bigg[ 
                    \gamma_k + \frac{\sigma_B^2}{\tau \rho}
                \bigg]
                e_{k3}
            +
                a d_k^2 \epsilon_k + a d_k^2 \epsilon_k e_{k4}^2
                +
                a d_k^2 e_{k4}^2
                +
                a \frac{\sigma_R^2}{\tau\rho} d_k e_{k4}^2 (\epsilon_k+1)
        \tag*{\normalsize (15)}
        \label{Theorem2Leak}
        \\
            &I_{ki}(\mathbf{A}, \mathbf{\Theta})
            =
                L^2 |f_{k}(\mathbf{A},\mathbf{\Theta})|^2 |f_{i}(\mathbf{A},\mathbf{\Theta})|^2
                c_k c_i \delta^2 \epsilon_k  \epsilon_i 
            +
                L |f_{k}(\mathbf{A},\mathbf{\Theta})|^2
                c_k \delta \epsilon_k
                \bigg[
                        c_{i} \big(
                        L (N-a) \delta
                    +
                        (N-a) \epsilon_i
                    +
                        (N-a)
                    +
                        2 L e_{k1}
                    \big)
                    +
                        \gamma_{i}
                \bigg]
            \notag
            \\
            &\ \ \ +
                L |f_{i}(\mathbf{A},\mathbf{\Theta})|^2
                c_i \delta \epsilon_i
                \bigg[ 
                    c_k e_{k2}
                    \big(
                            L (N-a) \delta e_{k2}
                        +
                            (N-a) \epsilon_k e_{k2}
                        +
                            (N-a) e_{k2}
                        +
                            2 L e_{k1}
                    \big)
                    +
                    \big(
                            \gamma_k
                        +
                            \frac{\sigma_B^2}{\tau \rho}
                    \big) e_{k2}^2
                \bigg]
            \notag
            \\
            &\ \ \ +
                    L^2 (N-a)^2 c_k c_i \delta^2 e_{k2}^2
                +
                    L (N-a)^2 c_k c_i 
                    \bigg[ 
                        \delta \big( 
                            \epsilon_k + \epsilon_i + 2
                        \big)
                        e_{k2}^2
                    +
                        (\epsilon_k + 1) (\epsilon_i + 1) e_{k3}
                    \bigg]
            +
                L^2 (N-a) c_k c_i e_{k1} 
                    \bigg[ 
                        (\epsilon_k + \epsilon_i + 1) e_{k1}
                        + 
                        2 \delta e_{k2}
                    \bigg]
            \notag
            \\
            &\ \ \ +
                L^2 c_k c_i \epsilon_k \epsilon_i e_{k1} 
                    \bigg[ 
                        |\overline{\mathbf{h}}_{k}^{H} 
                            (\mathbf{I}_{N} - \mathbf{A}^{H} \mathbf{A})
                            \overline{\mathbf{h}}_{i}
                        |^2 e_{k1}
            +
                        2 \delta {\rm{Re}} \bigg( 
                            f_{k}^{H} (\mathbf{A},\mathbf{\Theta})
                            f_{i}(\mathbf{A},\mathbf{\Theta})
                            \overline{\mathbf{h}}_i^{H}
                            (\mathbf{I} - \mathbf{A}^{H} \mathbf{A})
                            \overline{\mathbf{h}}_k
                        \bigg)
                    \bigg]
            +
                     L \gamma_i (\gamma_k + \frac{\sigma_B^2}{\tau \rho}) e_{k3}
            \notag
            \\
            &\ \ \ +
            L (N-a) 
                \bigg[ 
                    \big( \gamma_k + \frac{\sigma_B^2}{\tau \rho } \big)
                    c_i
                    \big(
                        \delta e_{k2}^2 + (\epsilon_i + 1) e_{k3}
                    \big)
                +
                    \gamma_i c_{k}
                    \big(
                        \delta e_{k2}^2 + (\epsilon_k + 1) e_{k3}
                    \big)
                \bigg]
            +
                d_{k} d_{i}
                        \epsilon_{k}\epsilon_{i}
                        \big|
                            g_{k,i} (\mathbf{A})
                        \big|^2
            \notag
            \\
            &\ \ \ +
                a d_{i} \bigg[ 
                    d_k \epsilon_k + 
                        (\epsilon_i + 1) 
                        (d_k + \frac{\sigma_R^2}{\tau\rho})
                        e_{k4}^2
                \bigg]
            +
                 2 {\rm{Re}} \bigg(
                        L \sqrt{c_k c_i d_k d_i} \delta \epsilon_k \epsilon_i 
                        f_{k}^{H}(\mathbf{B}) f_{i}(\mathbf{B}) g_{i,k}(\mathbf{A})
                        +
                        L \sqrt{c_k c_i d_{k} d_{i}} \epsilon_k \epsilon_i
                          e_{k1}
                         \big(\overline{\mathbf{h}}_k^{H}\overline{\mathbf{h}}_i - g_{k,i}(\mathbf{A}) \big) g_{i,k}(\mathbf{A})
                        \bigg)
            \tag*{\normalsize (16)}            
            \label{Theorem2Interference}
           \end{align}
        \setcounter{equation}{\value{term}}
        \hrulefill
        \end{figure*}
        \addtocounter{equation}{4}
        \vspace{-0.2cm}

With the SINR analysis in \textbf{Theorem 2}, the achievable rate can finally be derived in closed-form as (\ref{Equation_Imperfect_CSI}). The closed-form expression of the rate does not require the calculation of inverse matrices as well as numerical integrals. In contrast to time-consuming Monte Carlo simulations, the analytical rate expression provides a simple yet efficient way for performance evaluation. Moreover, the derived achievable rate covers the existing results of the following systems as special cases.

    \begin{itemize}
            \item For $a = 0$, the RDARS-aided system reduces to the RIS-aided system and the rate (\ref{Equation_Imperfect_CSI}) reduces to that of RIS-aided MIMO systems with imperfect CSI obtained by LMMSE channel estimator as in \cite{Kangda_Zhi}.
            \item For $a = N$, the RDARS-aided system becomes a DAS with two distributed antenna sets where one set is located at the BS with $L$ antennas and the other deployed at the remote location with $N$ antennas. The rate in (\ref{Equation_Imperfect_CSI}) then reduces to that of the DAS as in \cite{DAS}.
            \item For $N = a = 0$, the RDARS-aided system reduces to a traditional massive MIMO system and the rate in (\ref{Equation_Imperfect_CSI}) reduces to that of massive MIMO systems as given in \cite{SIG-093}.
    \end{itemize}

     The closed-form rate expression not only explicitly characterizes the impacts of the system parameters including $L$, $N$, AoA and AoD, path-loss parameters, and Rician factors on the performance as demonstrated in Section VI, it also enables effective TTS transceiver optimization which will be investigated in the next section. Notice that the achievable rate depends on the statistical CSI only. Then the RDARS phase shifts and the transmit powers can be jointly designed based on the S-CSI, which thus greatly alleviates the computational complexity and reduces the feedback overhead for RDARS configuration and power control, especially when a large scale RDARS is applied.
    

\vspace{-5pt}
\section{Joint Optimization of Uplink Transmit Power and RDARS}

    In this section, we jointly design the users' transmit power and the RDARS to maximize the uplink performance. Particularly, we consider the weighted sum rate as the objective function for system design. As introduced, the RDARS is a reconfigurable combination of DAS and RIS. Roughly speaking, it may provide distribution gain from DAS, reflection gain from RIS, as well as selection gain from the mode selection of each RDARS element. The mode selection for each RDARS element is implemented through the indicator matrix $\mathbf{A}$. This paper is the first work on RDARS-aided multi-user systems. For simplicity, here we mainly focus on the joint optimization of uplink transmit powers and the phase shifts for RDARS elements performing \textit{reflection mode} to exploit both distribution gain and reflection gain, by considering a fixed indicator matrix $\mathbf{A}$ (and thus $a$). We expect that the mode selection, i.e., the design of $\mathbf{A}$, will provide an extra degree of freedom in optimization for RDARS-aided systems and it will be investigated in the future. Given the indicator matrix $\mathbf{A}$, the corresponding term in ${\rm{SINR}}_{k}$, i.e., $E_{k}^{signal}(\mathbf{A},\mathbf{\Theta})$, $E_{k}^{leak}(\mathbf{A},\mathbf{\Theta})$, $I_{ki}(\mathbf{A},\mathbf{\Theta})$, $E_{k}^{noise}(\mathbf{A},\mathbf{\Theta})$ are concisely rewritten as $E_{k}^{signal}(\boldsymbol{\theta})$, $E_{k}^{leak}(\boldsymbol{\theta})$, $I_{ki}(\boldsymbol{\theta})$, $E_{k}^{noise}(\boldsymbol{\theta})$, respectively.

    
    
    
\vspace{-7pt}
\subsection{Weighted Sum Rate Maximization}

    Without loss of generality, the ratio $(\tau_c - \tau)/\tau$ is omitted for brevity. It will not affect the optimal solution for weighted sum rate maximization. Accordingly, the weighted sum rate maximization problem is formulated as 
    \begin{align}
    \begin{footnotesize}
            \begin{aligned}
            \mathcal{P}_o: \ \max\limits_{\boldsymbol{\theta}, \mathbf{p}}& \quad 
            f_o(\boldsymbol{\theta},\mathbf{p}) = \sum_{k=1}^{K} w_k R_k, \notag \\ 
                s.t. \ 
                &{\rm{CR1}}: \ |\boldsymbol{\theta}(i)| = 1, \forall i \in \widetilde{\mathcal{A}}, \\
                &{\rm{CR2}}: \ 0 \leq p_k \leq p_{\max}, \forall k \in \mathcal{K},
            \end{aligned}
    \end{footnotesize}
    \end{align}
   where $\mathbf{p} = [p_1, p_2, \cdots, p_K]^{T}$, $w_k$ is the weighting parameter which can be used to model the priority of the $k^{th}$ user and $R_k$ is given in \eqref{Equation_Imperfect_CSI} with ${\rm{SINR}}_k$ derived in (\ref{SINR_Theorem2}). ${\rm{CR1}}$ represents the unit-modulus constraint of the phase shifts for RDARS elements performing \textit{reflection mode}, and ${\rm{CR2}}$ refers to the power constraint of each user's transmit power. 

    The problem $\mathcal{P}_o$ is difficult to tackle due to the coupling of variables, e.g., the transmit powers ${p}_k, \forall k \in \mathcal{K}$ and the phase shifts $\boldsymbol{\theta}$ are deeply coupled in the complicated expression of the SINR as shown in (\ref{SINR_Theorem2}). 
    Besides, the involvement of the fourth order terms of the phase shifts in the desired signal power \ref{Theorem2Signal} and the multi-user interference power \ref{Theorem2Interference}, and the non-convex unit-modulus constraints for $\boldsymbol{\theta}$ further hinder the optimal solution. \textit{This optimization problem is challenging and has not been considered even for the special case of RDARS-aided systems, i.e., the RIS-aided systems by setting $a = 0$.}\footnote{TTS transceiver design for RIS-aided systems with imperfect CSI and MRC was investigated in \cite{Kangda_Zhi}. However, the design aimed at the maximization of the minimum rate and only the optimization of the phase shifts was considered.} To deal with it, we first transform $\mathcal{P}_0$ into a more tractable form as $\mathcal{P}_1$. Then, the BCD method is exploited to find a stationary solution. Closed-form update rules for the transmit power design $p_k$, as well as closed-form dual variables are derived. As for the design of phase shifts $\boldsymbol{\theta}$, two algorithms are proposed. The details of the proposed joint design algorithm are provided subsequently.
    
    
\vspace{-14pt}
\subsection{Problem Reformulation}

    In order to effectively solve the problem $\mathcal{P}_0$, we resort to the FP technique to reformulate the optimization problem. To be specific, two transformations are performed sequentially as follows.

\subsubsection{Lagragian Dual Transform}
    By defining a collection of auxiliary variables $\{ \eta_k \},\forall k \in \mathcal{K}$, the objective function of the original problem $\mathcal{P}_0$ can be equivalently transformed to its Lagragian dual function
    $f_r(\boldsymbol{\theta},\mathbf{p},\boldsymbol{\eta})$ as shown in \ref{f_r} at the top of this page, where $\boldsymbol{\eta} \triangleq [\eta_1, \eta_2, \cdots, \eta_K]^{T}$ with constraints $\eta_k \geq 0, \forall k \in \mathcal{K}$.

    \newcounter{f_r}
        \begin{figure*}[!htb]
        \footnotesize
        \setcounter{f_r}{\value{equation}}
        \setcounter{equation}{16}
        \begin{align} 
            f_r(\boldsymbol{\theta},\mathbf{p},\boldsymbol{\eta}) 
            &= 
                \sum_{k=1}^{K} 
            \bigg(
                w_k \log(1 + \eta_k)
                -  w_k \eta_k 
            +  
                    \frac{ w_k (1+\eta_k) p_k E_{k}^{signal} (\boldsymbol{\theta}) }
                    {
                        p_k E_{k}^{signal} (\boldsymbol{\theta}) 
                    + 
                        p_k E_{k}^{leak} (\boldsymbol{\theta})
                    + 
                        \sum_{\substack{i = 1 \\i \neq k}}^{K} p_i I_{k,i} (\mathbf{\theta})
                    +
                        E_{k}^{noise} (\boldsymbol{\theta})
                    }
            \bigg)
            \tag*{\normalsize (17)}
            \label{f_r}
        \end{align}
        \setcounter{equation}{\value{f_r}}
        \hrulefill
        \end{figure*}
        \addtocounter{equation}{1}

\subsubsection{Quadratic Transform}
    Next, the quadratic transform is conducted to tackle the remaining fractional terms in $f_r(\boldsymbol{\theta},\mathbf{p},\boldsymbol{\eta})$. Specifically, by introducing another set of auxiliary variables $\{ \chi_k\},\forall k \in \mathcal{K}$ for each ratio in the last term of \ref{f_r}, the objective function $f_r(\boldsymbol{\theta},\mathbf{p},\boldsymbol{\eta})$ can be further reformulated as $f_q(\boldsymbol{\theta},\mathbf{p},\boldsymbol{\eta},\boldsymbol{\chi})$ as given in \ref{f_q} at the top of this page, where $\boldsymbol{\chi} \triangleq [\chi_1, \chi_2, \cdots, \chi_K]^{T}$ with constraints $\chi_k \geq 0, \forall k \in \mathcal{K}$.

    \newcounter{f_q}
        \begin{figure*}[!htb]
        \footnotesize
        \setcounter{f_q}{\value{equation}}
        \setcounter{equation}{17}
        \begin{align}
            &f_q(\boldsymbol{\theta},\mathbf{p},\boldsymbol{\eta},\boldsymbol{\chi})
            \notag
            \\
            &=
            \sum_{k=1}^{K} \bigg( 
                w_k \log(1 + \eta_k) - w_k \eta_k
                +
                2 \chi_k \sqrt{ w_k (1+\eta_k) p_k E_{k}^{signal} (\boldsymbol{\theta})}
            -
                \chi_k^2 \big[ 
                     p_k E_{k}^{signal} (\boldsymbol{\theta}) 
                    + 
                        p_k E_{k}^{leak} (\boldsymbol{\theta})
                    + 
                        \sum_{\substack{i = 1 \\i \neq k}}^{K} p_i I_{k,i} (\boldsymbol{\theta})
                    +
                        E_{k}^{noise} (\boldsymbol{\theta})
                \big]
            \bigg)
             \tag*{\normalsize (18)}
            \label{f_q}
        \end{align}
        \setcounter{equation}{\value{f_q}}
        \hrulefill
        \end{figure*}
        \addtocounter{equation}{1}

    After applying the above FP approach, the original problem has been equivalently transformed to the following problem
    \begin{align}
        \begin{footnotesize}
            \begin{aligned} 
            \mathcal{P}_1: \ \max\limits_{\boldsymbol{\theta},\mathbf{p},\boldsymbol{\eta},\boldsymbol{\chi}}& \quad 
            f_q(\boldsymbol{\theta},\mathbf{p},\boldsymbol{\eta},\boldsymbol{\chi}), \notag \\ 
                s.t. \ 
                &{\rm{CR1}}: \ |\boldsymbol{\theta}(i)| = 1  , \forall i\in \widetilde{\mathcal{A}}, \\ 
                &{\rm{CR2}}: \ 0 \leq p_k \leq p_{\max}, \forall k \in \mathcal{K}, \\
                &{\rm{CR3}}: \ 0 \leq \eta_k , \ 0 \leq \chi_k, \forall k \in \mathcal{K}.
            \end{aligned}
        \end{footnotesize}
    \end{align}
    It should be noted that the above transformations decouple the numerator and denominator of each ratio term and move the ${\rm{SINR}}_k$ outside the logarithm of the rate expression in (\ref{Equation_Imperfect_CSI}) \cite{Kaiming_Shen}. However, finding a globally optimal solution for $\mathcal{P}_1$ is still intractable due to the non-convex constraints imposed by $\boldsymbol{\theta}$ in CR1. As such, we propose to decouple these variables into four disjoint blocks and then the non-convex BCD method is invoked to update these variables in a cyclically iterative manner under the guarantee of stationary solution for $\mathcal{P}_1$.

\vspace{-10pt}   
\subsection{Update of Dual Variables $\boldsymbol{\eta}$ and $\boldsymbol{\chi}$}

    Following the BCD method, the optimal dual variables $\boldsymbol{\eta}$ and $\boldsymbol{\chi}$ in each iteration can be found by setting the corresponding first order derivatives of the objective function to zero. To be specific, given ${\boldsymbol{\theta}}$, ${p}_k$, and ${\chi}_k, \forall k \in \mathcal{K}$ obtained in the last iteration, by setting the first order derivative of the objective function \ref{f_q} w.r.t. $\eta_k$ to zero and after some simple manipulations, its optimal solution is derived as
    \begin{align} \label{Optimal_eta}
    \begin{footnotesize}
    \begin{aligned}
        \eta_k^{opt} = 
                    \frac{ \overline{{\kappa}}_{k}^2 + \overline{{\kappa}}_{k} \sqrt{ \overline{{\kappa}}_{k}^2 + 4  }}{2},
    \end{aligned}
    \end{footnotesize}
    \end{align}
     where $\overline{{\kappa}}_{k}$ is defined as
    \begin{align}
    \begin{footnotesize}
    \begin{aligned}
        \overline{{\kappa}}_{k} = \frac{1}{\sqrt{w}_k} \sqrt{ {\chi}_{k}^2 {p}_k E_{k}^{signal} ({\boldsymbol{\theta}} ) }.
    \end{aligned}
    \end{footnotesize}
    \end{align}
    Similarly, the optimal solution of $\chi_k$ is obtained as $\chi_k^{opt}$
    \begin{align} \label{Optimal_chi}
    \begin{footnotesize}
    \begin{aligned}
            &= \frac{\sqrt{ w_k (1+{\eta}_k) {p}_k E_{k}^{signal} ({\boldsymbol{\theta}})}}
            {  
                     {p}_k E_{k}^{signal} ({\boldsymbol{\theta}}) 
                    + 
                        {p}_k E_{k}^{leak} ({\boldsymbol{\theta}})
                    + 
                        \sum_{i = 1 ,i \neq k}^{K} {p}_i I_{k,i} ({\boldsymbol{\theta}})
                    +
                        E_{k}^{noise} ({\boldsymbol{\theta}})
            }.
    \end{aligned}
    \end{footnotesize}
    \end{align}

\vspace{-10pt}
\subsection{Design of Transmit Powers $\mathbf{p}$}

    Given ${\boldsymbol{\theta}},{\eta}_k, {\chi}_k, \forall k \in \mathcal{K}$, the optimal transmit power $p_k^{opt}$ is derived by setting $\frac{\partial f_q}{\partial p_k} = 0$, yielding $p_k^{opt} = \min ( p_{\max}, \widetilde{p}_{k} )$, where
    \begin{align} \label{Optimal_power}
    \begin{footnotesize}
    \begin{aligned}
        \widetilde{p}_{k} = \frac{
            w_k(1+{\eta}_k) E_k^{signal}({\boldsymbol{\theta}}) {\chi}_{k}^2
        }
        {
            \big[ 
                {\chi}_{k}^2 \big(
                    E_k^{signal}({\boldsymbol{\theta}})+
                    E_k^{leak}({\boldsymbol{\theta}})
                \big)
                +
                \sum_{\substack{i = 1 \\ i \neq k}}^{K} {\chi}_{i}^2 
                    I_{i,k}({\boldsymbol{\theta}})
            \big]^2
        }.
    \end{aligned}
    \end{footnotesize}
    \end{align}

\vspace{-10pt}
\subsection{Design of Phase Shifts $\boldsymbol{\theta}$}

    The optimization of the phase shifts is more challenging than the design of other variables due to the unit-modulus constraints and the high order involvement. Here two different algorithms for the design of $\boldsymbol{\theta}$ are proposed. Particularly, the first proposed algorithm, i.e., the ``Riemannian Gradient Ascent algorithm", adopts the traditional RGA algorithm by performing the line search on manifolds \cite{Alhujaili}, while the second algorithm, i.e., the ``Two-Tier MM-based algorithm", leverages the MM framework to find a closed-form optimal solution for $\boldsymbol{\theta}$.

    \vspace{-5pt}
    \subsubsection{Riemannian Gradient Ascent-Based Algorithm}
    Since only the phases of the elements performing \textit{reflection mode} need to be optimized, the search space of the phase shifts is the product of $N-a$ complex circles. Denoting such search space as $\mathcal{S}^{N-a} = { \boldsymbol{\theta}_{\widetilde{\mathcal{A}}} \in \mathbb{C}^{N-a}: |\boldsymbol{\theta}_{\widetilde{\mathcal{A}}}(i)| = 1, \forall i \in \widetilde{\mathcal{A}} }$, then $\mathcal{S}^{N-a}$ is known as the \textit{complex circle manifold} (CCM). As such, the optimization w.r.t. $\boldsymbol{\theta}_{\widetilde{\mathcal{A}}}$ can be solved by the CCM algorithm which has been widely applied for radar beampattern design \cite{Alhujaili} as well as phase shift design for RIS-aided system \cite{Cunhua_Pan_Overview}. Conceptually, the CCM algorithm can be treated as a class of line-search algorithms for unconstrained nonlinear optimization in $\mathbb{C}^{N-a}$ with additional steps to enforce the obtained ascent result still on the manifold. Specifically, setting $\nabla^{t} f_q$ as the Euclidean gradient of $\boldsymbol{\theta}_{\widetilde{\mathcal{A}}}^{t}$ w.r.t. the objective function, the Riemannian gradient, i.e., $\nabla_{R}^{t} f_q$, is then obtained by projecting the Euclidean gradient $\nabla^{t} f_q$ onto the tangent space of the CCM at the current point $\boldsymbol{\theta}_{\widetilde{\mathcal{A}}}^{t}$, resulting in $\nabla_{R}^{t} f_q = \nabla^{t} f_q - {\rm{Re}}\{ {\nabla^{t}} f_q \odot {\boldsymbol{\theta}_{\widetilde{\mathcal{A}}}^{t}}^{*} \} \odot \boldsymbol{\theta}_{\widetilde{\mathcal{A}}}^{t}$, where $\boldsymbol{\theta}$ is derived as in the following lemma. 
    \begin{lemma}
        The Euclidean gradient w.r.t. $\boldsymbol{\theta}_{\widetilde{\mathcal{A}}}$ is derived as \ref{RGA_gradient}, where $\widetilde{\mathbf{J}} = [(\mathbf{C}^{*} \otimes \mathbf{C})^{H}\mathbf{J}(\mathbf{C}^{*} \otimes \mathbf{C})]_{\widetilde{\mathcal{A}}_{{\rm{kron}}}}$, $\widetilde{\mathbf{K}}=[\mathbf{C}^{H}\mathbf{K}\mathbf{C}]_{\widetilde{\mathcal{A}}}$, $\mathbf{C} \triangleq \mathbf{I}-\mathbf{A}^{H}\mathbf{A} \in \mathbb{C}^{N \times N}$, and  $[\cdot]_{\widetilde{\mathcal{A}}_{{\rm{kron}}}}$ is the operator which returns the sub-matrix by taking the rows and columns of the input matrix specified by the index set corresponding to the non-zero elements of ${\rm{vec}}(\mathbf{C}^{*} \otimes \mathbf{C})$. Similarly, $[\cdot]_{\widetilde{\mathcal{A}}}$ is the operator which returns the sub-matrix by taking the rows and columns of the input matrix corresponding to the non-zero elements of ${\rm{vec}}(\mathbf{C})$. The definitions of $\mathbf{J}$ and $\mathbf{K}$ are given as in \ref{J_matrix} and \ref{K_matrix} at the top of this page, respectively.
        
    \end{lemma}
    \begin{proof}
        See Appendix D for proof and the definitions of various parameters.
    \end{proof}

    \newcounter{gradient}
        \begin{figure*}[!htb]
        \footnotesize
        \setcounter{gradient}{\value{equation}}
        \setcounter{equation}{22}
        \begin{align}
            \nabla f_q  
            &=
            2\bigg[ 
                (\boldsymbol{\theta}_{\widetilde{\mathcal{A}}}^{*}\otimes\boldsymbol{\theta}_{\widetilde{\mathcal{A}}})^{T}\widetilde{\mathbf{J}}^{T}(\mathbf{I}\otimes\mathbf{I})(\boldsymbol{\theta}_{\widetilde{\mathcal{A}}}\otimes\mathbf{I})
                +
                (\boldsymbol{\theta}_{\widetilde{\mathcal{A}}}^{*}\otimes\boldsymbol{\theta}_{\widetilde{\mathcal{A}}})^{H}\widetilde{\mathbf{J}}(\mathbf{I}\otimes\mathbf{I})(\mathbf{I}\otimes\boldsymbol{\theta}_{\widetilde{\mathcal{A}}})
                +
                \boldsymbol{\theta}_{\widetilde{\mathcal{A}}}^{T}\widetilde{\mathbf{K}}^{T}
        \bigg]^{T}
        \tag*{\normalsize (23)}
        \label{RGA_gradient}
        \end{align}
        \setcounter{equation}{\value{gradient}}
        \hrulefill
        \end{figure*}
        \addtocounter{equation}{1}

    \begin{figure*}[!htb]
      \footnotesize
        \setcounter{gradient}{\value{equation}}
        \setcounter{equation}{24}
        \begin{align}
            \mathbf{J}  
                &=
                     \bigg\{
                    \sum_{k=1}^{K} \sum_{i=1, i \neq k}^{K}
                        -\chi_k^2 p_k
                    \big[ 
                         s_{k,1}^2 (\mathbf{C}_k^T \otimes \mathbf{C}_k)
                    \big]
                        -\chi_k^2 p_i
                    \big[ 
                         u_{k,i,1} (\mathbf{C}_i^T \otimes \mathbf{C}_k)
                         +
                         u_{k,i,4} (\mathbf{F}_{k,i} \otimes \mathbf{F}_{k,i}^{H})
                    +
                        u_{k,i,5}  \big(
                                    (\mathbf{F}_{k,i}^{*} \otimes \mathbf{D}_{k,i})
                                    +
                                    (\mathbf{F}_{k,i}^{*} \otimes \mathbf{D}_{k,i})^{H}
                        \big) \notag
                \label{J_matrix}
                \\
                &+
                         u_{k,i,6} (\mathbf{F}_{k,i}^{H} \otimes \mathbf{F}_{k,i})
                    -
                         u_{k,i,7}  \big( 
                                        (\mathbf{F}_{k,i}^{*} \otimes \mathbf{D}_{k,i}) 
                                    +
                                        (\mathbf{F}_{k,i}^{*} \otimes \mathbf{D}_{k,i})^{H} 
                        \big)
                    -
                         u_{k,i,8} \big( 
                                    (\mathbf{F}_{k,i}^{*} \otimes \mathbf{F}_{k,i}) 
                                    +
                                    (\mathbf{F}_{k,i}^{*} \otimes \mathbf{F}_{k,i})^{H}
                                    \big)
                    \big]
                    \bigg\}   \tag*{\normalsize (24)}
                \\
            \mathbf{K}  
                &=
                    \sum_{k=1}^{K} \sum_{i=1, i \neq k}^{K}
                \bigg\{
                        2 \chi_k \sqrt{w_k(1+\eta_k)p_k} s_{k,1} \mathbf{C}_k
                    -
                        \chi_k^2 p_k 2 s_{k,1} s_{k,2} \mathbf{C}_k
                    -
                        \chi_k^2 p_k l_{k,1} \mathbf{C}_k
                    -
                        \chi_k^2 n_{k,1} \mathbf{C}_k
                -
                    \chi_k^2 p_i \bigg[ 
                        u_{k,i,2}\mathbf{C}_k
                        +
                        u_{k,i,3}\mathbf{C}_i
                        \notag
                \label{K_matrix}  
                \\
                &-
                        u_{k,i,6}  \big( m_{k,i}^{H} \mathbf{F}_{k,i})
                        -
                        u_{k,i,6}  \big( m_{k,i}^{H} \mathbf{F}_{k,i})^{H}
                +
                        u_{k,i,7} \big( m_{k,i}^{H} \mathbf{D}_{k,i} \big)
                        +
                        u_{k,i,7} \big( m_{k,i}^{H} \mathbf{D}_{k,i} \big)^{H}
                        +
                        u_{k,i,8} \big( m_{k,i}^{H} \mathbf{F}_{k,i} \big)
                        +
                        u_{k,i,8} \big( m_{k,i}^{H} \mathbf{F}_{k,i} \big)^{H}
                    \bigg]
                    \bigg\} \tag*{\normalsize (25)}
            \end{align}
            \setcounter{equation}{\value{gradient}}
            \hrulefill
    \end{figure*}    
    \addtocounter{equation}{2}
        
Armed with the above Euclidean gradient $\nabla f_q$, the update of  $\boldsymbol{\theta}_{\widetilde{\mathcal{A}}}^{t}$ along the Riemannian gradient is obtained as 
    \begin{align} \label{decent}
    \begin{footnotesize}
    \begin{aligned}
        \overline{\boldsymbol{\theta}}_{\widetilde{\mathcal{A}}}^{t+1} = \boldsymbol{\theta}_{\widetilde{\mathcal{A}}}^{t} + \rho_t \nabla_{R}^{t} f_q,
    \end{aligned}
    \end{footnotesize}
    \end{align}
where $\rho_t$ refers to the step size. To ensure the ascent is on the CCM, a \textit{retraction} step is required to obtain the ascend result as 
    \begin{align} \label{retraction}
    \begin{footnotesize}
    \begin{aligned}
        \overline{\boldsymbol{\theta}}_{\widetilde{\mathcal{A}}}^{t+1} = \overline{\boldsymbol{\theta}}_{\widetilde{\mathcal{A}}}^{t} / |\overline{\boldsymbol{\theta}}_{\widetilde{\mathcal{A}}}^{t}|.
    \end{aligned}
    \end{footnotesize}
    \end{align}
To preserve the monotonicity of the RGA algorithm, the backtracking line search is adopted to determine the step size $\rho_t$ \cite{Jintao_Wang}. The detail of this RGA algorithm is summarized in \textbf{Algorithm 1}.

    \vspace{-5pt}
    \begin{algorithm} \label{RGA_algorithm}
    \renewcommand{\algorithmicrequire}{\textbf{Input:}}
    \renewcommand{\algorithmicensure}{\textbf{Output:}}
	\caption{RGA-based Algorithm}
	\label{alg2}
	   \begin{algorithmic}[1]
	   \footnotesize
        		\REQUIRE The optimal $\boldsymbol{\theta}_n$, $\mathbf{p}_n$, $\boldsymbol{\eta}_n$ and $\boldsymbol{\chi}_n$ at the $n^{th}$ iteration of BCD; Maximum iteration $t_{\max}$.
                \ENSURE  $\boldsymbol{\theta}_{n+1}$
                \STATE Initialization: $t = 0$, $c \in (0,1)$, $\boldsymbol{\theta}^t = \boldsymbol{\theta}_n$.
        		\REPEAT
        		\STATE Update $\boldsymbol{\theta}^{t+1}$ by (\ref{retraction});
                \WHILE {$f_q(\boldsymbol{\theta}^{t+1},\mathbf{p},\boldsymbol{\eta},\boldsymbol{\chi})<f_q(\boldsymbol{\theta}^{t},\mathbf{p},\boldsymbol{\eta},\boldsymbol{\chi})$}
                    \STATE $\rho_{t} = c\rho_{t}$;
                    \STATE Update $\boldsymbol{\theta}^{t+1}$ by (\ref{decent}).
                \ENDWHILE
                \STATE $t = t+1$;
        		\UNTIL  $f_q(\boldsymbol{\theta},\mathbf{p},\boldsymbol{\eta},\boldsymbol{\chi})$ in \ref{f_q} converges or $t > t_{\max}$;
                \RETURN $\boldsymbol{\theta}_{n+1} = \boldsymbol{\theta}^{t}$.
    	\end{algorithmic}
    \end{algorithm}

    \subsubsection{Two-Tier MM-based Algorithm}

    Usually, closed-form solution is more attractive for practical applications. To achieve that, we further analyze the inherent structure of the objective function and resort to the MM framework for the optimal design. By utilizing the property ${\rm{Tr}}(\mathbf{X}\mathbf{Y}\mathbf{Z}\mathbf{W}) = {\rm{vec}}^{H}(\mathbf{X}^{H})(\mathbf{W}^{T}\otimes\mathbf{Y}){\rm{vec}}(\mathbf{Z})$, the objective function in \ref{f_q} can be simplified as
    \begin{equation}
        \begin{footnotesize}
        \begin{aligned}\label{f_q_1_2}
        f_q(\boldsymbol{\theta},\mathbf{p},\boldsymbol{\eta},\boldsymbol{\chi}) 
        &=
            \mathbf{v}^{H}
                \mathbf{J}_1
            \mathbf{v}
            +
                {\rm{vec}}^{H}(\boldsymbol{\Theta})
                \widetilde{\mathbf{K}}_1
                {\rm{vec}}(\boldsymbol{\Theta})
        \\
        &+
                2{\rm{Re}} \big\{ 
                    {\rm{vec}}^{H}(\mathbf{\Theta})
                    \widetilde{\mathbf{K}}_2
                    {\rm{vec}}(\mathbf{\Theta})
                \big\}
            +
                {\rm{const}}1,
        \end{aligned}
        \end{footnotesize}
    \end{equation}
    \normalsize
    
    \noindent where $\mathbf{v} \triangleq (\mathbf{\Theta}\otimes\mathbf{\Theta}^{H}) {\rm{vec}}(\mathbf{S}) \in \mathbb{C}^{N^2 \times 1}$, $\mathbf{S} \triangleq \mathbf{C}^{H}\mathbf{a}_N\mathbf{a}_N^{H}\mathbf{C} \in \mathbb{C}^{N \times N}$,
    $\widetilde{\mathbf{K}}_1 = (\mathbf{I}\otimes\mathbf{C})^{H} \mathbf{K}_1 (\mathbf{I}\otimes\mathbf{C}) \in \mathbb{C}^{N^2 \times N^2}$, $\widetilde{\mathbf{K}}_2 = (\mathbf{I}\otimes\mathbf{C})^{H} \mathbf{K}_2 (\mathbf{I}\otimes\mathbf{C}) \in \mathbb{C}^{N^2 \times N^2}$, and ${\rm{const}}1$ is a constant irrelevant to $\boldsymbol{\theta}$. The definitions of $\mathbf{J}_1 \in \mathbb{C}^{N^2 \times N^2}$, $\mathbf{K}_1 \in \mathbb{C}^{N^2 \times N^2}$, and $\mathbf{K}_2 \in \mathbb{C}^{N^2 \times N^2}$ are given in \ref{J_matrix_1}, \ref{K_matrix_1}, \ref{K_matrix_2} at the top of next page, respectively. The constants such as $s_{i}$ and $l_{k,i}$ can be derived by direct inspection of \textbf{Theorem 2} and thus omitted here for brevity.

        \newcounter{JK1K2}
        \begin{figure*}[!htb]
        \setcounter{JK1K2}{\value{equation}}
        \setcounter{equation}{28}
        \footnotesize
        \begin{align} 
        \mathbf{J}_1 
        &=
            \bigg\{
                \sum_{k=1}^{K} \sum_{i=1, i \neq k}^{K}
                    -\chi_k^2 p_k
                \big[ 
                    s_{k,1}^2 \big( (\overline{\mathbf{h}}_k \overline{\mathbf{h}}_k^{H})^{T} \otimes      (\overline{\mathbf{h}}_k\overline{\mathbf{h}}_k^{H}) \big)
                \big]
                    -\chi_k^2 p_i
                \big[
                    u_{k,i,1} \big( (\overline{\mathbf{h}}_k\overline{\mathbf{h}}_k^{H})^{T} \otimes (\overline{\mathbf{h}}_i\overline{\mathbf{h}}_i^{H}) \big)
                \big]
                \bigg\}
        \tag*{\normalsize (29)}
        \label{J_matrix_1}
            \\
        \mathbf{K}_1 
            &=
                \sum_{k=1}^{K} \sum_{i=1, i \neq k}^{K}
            \bigg\{
                    \big[ 
                         2 \chi_k \sqrt{w_k(1+\eta_k)p_k} s_{k,1}
                         -
                         \chi_k^2(p_kl_{k,1} + n_{k,1} + p_iu_{k,i,2})
                    \big]
                    \big[   
                            (\overline{\mathbf{h}}_k\overline{\mathbf{h}}_k^{H})^{T} \otimes                
                            (\mathbf{a}_N\mathbf{a}_N^{H})
                    \big]
            -
                    \chi_k^2p_i u_{k,i,3}
                    \big[   
                            (\overline{\mathbf{h}}_i\overline{\mathbf{h}}_i^{H})^{T} \otimes                
                            (\mathbf{a}_N\mathbf{a}_N^{H})
                    \big]
            \bigg\}  
        \tag*{\normalsize (30)}
        \label{K_matrix_1}
        \\
        \mathbf{K}_2
            &= 
                \sum_{k=1}^{K} \sum_{i=1, i \neq k}^{K}
            \bigg\{
                   -\chi_k^2 p_k
                    \big[   
                            s_{k,1} s_{k,2}
                            (\overline{\mathbf{h}}_k\overline{\mathbf{h}}_k^{H})^{T}
                                \otimes
                            (\mathbf{a}_N\mathbf{a}_N^{H})
                    \big]
            -\chi_k^2 p_i
                    \big[
                            u_{k,i,5}(\overline{\mathbf{h}}_i\overline{\mathbf{h}}_{i}^{H}\mathbf{C}\overline{\mathbf{h}}_k\overline{\mathbf{h}}_k^{H})^{T}
                            \otimes
                            (\mathbf{a}_N\mathbf{a}_N^{H})
                           + u_{k,i,7}
                            (\overline{\mathbf{h}}_i\overline{\mathbf{h}}_i^{H}(\mathbf{I}-\mathbf{C}^{H})\overline{\mathbf{h}}_k\overline{\mathbf{h}}_k^{H})^{T}
                            \otimes
                            (\mathbf{a}_N\mathbf{a}_N^{H})
                    \big]
                \bigg\}
        \tag*{\normalsize (31)}
        \label{K_matrix_2} 
        \end{align}
        \setcounter{equation}{\value{JK1K2}}
        \hrulefill
        \end{figure*}
        \addtocounter{equation}{3}


The objective function in (\ref{f_q_1_2}) is still complicated with not only the second-order term, i.e., $\mathbf{K}_1$ and $\mathbf{K}_2$, but also the \textit{``fourth-order"} term, i.e., $\mathbf{J}_1$. To proceed, we further explore the inherent structure of \eqref{f_q_1_2} to simplify the objective function. Specifically, the objective function can be first simplified by exploring the sparse structure of ${\rm{vec}}(\mathbf{\Theta})$. Let $\mathbf{K}_3 \in \mathbb{C}^{N \times N}$ and $\mathbf{K}_4 \in \mathbb{C}^{N \times N}$ be the matrices made up of all the elements of the $[(i-1)N + i]^{th}, \forall i \in [N]$ columns and $[(j-1)N + j]^{th}, \forall j \in [N]$ rows of $\mathbf{K}_1$ and $\mathbf{K}_2$ respectively. The selected rows and columns in fact correspond to the sparse non-zero values of ${\rm{vec}}(\mathbf{\Theta})$. Define $\hat{\mathbf{J}}_1 = {\rm{diag}}(\mathbf{s}^{H})\mathbf{J}_1{\rm{diag}}(\mathbf{s})  \in \mathbb{C}^{N^2 \times N^2}$, where $\mathbf{s} = [\mathbf{s}_1^{T}, \cdots, \mathbf{s}_N^{T}]^{T} \in \mathbb{C}^{N^2 \times 1}$ and $\mathbf{s}_i \in \mathbb{C}^{N \times 1}$ is the $i^{th}$ column of $\mathbf{S}$. We also define $\mathbf{J}_2 \in \mathbb{C}^{(N+N^2) \times (N+N^2)}$ as 
    \begin{align}
    \begin{footnotesize}
    \begin{aligned}
        \mathbf{J}_2
            =
            \begin{bmatrix}
                -\mathbf{K}_3 &\mathbf{0} \\
                \mathbf{0}    & -\hat{\mathbf{J}}_1
            \end{bmatrix}.
    \end{aligned}
    \end{footnotesize}
    \end{align}
Armed with the above auxiliary variables, the optimization of $\boldsymbol{\theta}$ can be equivalently transformed as
    \begin{align}
    \begin{footnotesize}
            \begin{aligned} 
            \mathcal{P}_{1}^{sub}: \ \min\limits_{\boldsymbol{\theta}}& \quad  
            g_0(\boldsymbol{\theta}) 
            \triangleq
            \widetilde{\boldsymbol{\theta}}^{H} \mathbf{J}_2 \widetilde{\boldsymbol{\theta}}
            -
            2{\rm{Re}}\big\{\boldsymbol{\theta}^{H}\mathbf{K}_4\boldsymbol{\theta}\big\}
            +{\rm{const}}2, 
            \notag
             \\ 
                s.t. \ 
                &{\rm{CR1}}: \ |\boldsymbol{\theta}(i)| = 1 , \forall i \in \widetilde{\mathcal{A}},
            \end{aligned}
    \end{footnotesize}
    \end{align}
where $g_0(\boldsymbol{\theta}) = -f_q(\boldsymbol{\theta},\mathbf{p},\boldsymbol{\eta},\boldsymbol{\chi})$ and $\widetilde{\boldsymbol{\theta}} = [\boldsymbol{\theta}^{T}, \boldsymbol{\theta}^{T} \otimes \boldsymbol{\theta}^{H}]^{T} \in \mathbb{C}^{(N+N^{2})\times 1}$. Notice that $g_0(\boldsymbol{\theta})$ is still non-convex w.r.t. $\boldsymbol{\theta}$, and the coupling between $\widetilde{\boldsymbol{\theta}}$ and $\boldsymbol{\theta}$ hinders the optimal solution. To tackle such a challenging problem, a two-tier MM technique is applied which admits a closed-form solution in each iteration.

Notably, $\mathbf{J}_2$ is a Hermitian matrix, i.e., $\mathbf{J}_2^{H} = \mathbf{J}_2$. As such, a surrogate function which is a locally tight upper bound of  $\widetilde{\boldsymbol{\theta}}^{H}\mathbf{J}_2\widetilde{\boldsymbol{\theta}}$ at given $\boldsymbol{\theta}_t$ (and equivalently given $\widetilde{\boldsymbol{\theta}}_t$) is derived as \cite[(Lemma 12)]{Ying_Sun}
    \begin{align} \label{MM_surrogate_1}
    \begin{footnotesize}
    \begin{aligned}
        \widetilde{\boldsymbol{\theta}}^{H}\mathbf{J}_2\widetilde{\boldsymbol{\theta}} 
        \leq
        \widetilde{\boldsymbol{\theta}}^{H}\mathbf{\Lambda}_1\widetilde{\boldsymbol{\theta}}
        +
        2 {\rm{Re}} \big\{ 
            \widetilde{\boldsymbol{\theta}}^{H}(\mathbf{J}_2 - \mathbf{\Lambda}_1) \widetilde{\boldsymbol{\theta}}_t
        \big\}
        -
        \widetilde{\boldsymbol{\theta}}_{t}^{H}
        (\mathbf{J}_2 - \mathbf{\Lambda}_1)
        \widetilde{\boldsymbol{\theta}}_t,
    \end{aligned}
    \end{footnotesize}
    \end{align}
where $\mathbf{\Lambda}_1 = \lambda_{max}(\mathbf{J}_2)\mathbf{I}$.
It is readily inferred that the first and the third terms in (\ref{MM_surrogate_1}) are constant given $\widetilde{\boldsymbol{\theta}}_t$. Then, a surrogate function for the objective function in $\mathcal{P}_{1}^{sub}$ at $\boldsymbol{\theta}_t$ is given by 
\begin{align}
\begin{footnotesize}
\begin{aligned}
    g_0(\boldsymbol{\theta}) \leq
    g_1(\boldsymbol{\theta};\boldsymbol{\theta}_{t})
            \triangleq
            2 {\rm{Re}} \big\{  
                \boldsymbol{\theta}^{H}\widetilde{\mathbf{u}}_t
            +
                \boldsymbol{\theta}^{T}\widetilde{\mathbf{V}}_t\boldsymbol{\theta}^{*}    
            \big\}+{\rm{const}}3,
\end{aligned}
\end{footnotesize}
\end{align}
where $\widetilde{\mathbf{u}}_t = [-\mathbf{K}_3 - \lambda_{\max}(\mathbf{J}_2)\mathbf{I}]
\boldsymbol{\theta}_t$, $\widetilde{\mathbf{V}}_t = \hat{\mathbf{V}}_t - \mathbf{K}_4^{T}$, ${\rm{vec}}(\hat{\mathbf{V}}_t) = \mathbf{v}_t$, and we have $\mathbf{v}_t = [-\hat{\mathbf{J}}_1 - \lambda_{\max}(\mathbf{J}_2)\mathbf{I}](\boldsymbol{\theta}_t \otimes \boldsymbol{\theta}_t^{*})$. 

The above process constructs an effective surrogate function
for $g_0(\boldsymbol{\theta})$ and completes the transformation from the implicit function to the explicit one w.r.t. $\boldsymbol{\theta}$. Nevertheless, it is clear that $g_1(\boldsymbol{\theta};\boldsymbol{\theta}_{t})$ is a quadratic function where a closed-form solution is still not available. As such, the MM technique is applied again to further construct a tractable surrogate function for $g_1(\boldsymbol{\theta};\boldsymbol{\theta}_t)$. Specifically, let $\overline{\mathbf{V}}_t \triangleq \widetilde{\mathbf{V}}_t^{*} + \widetilde{\mathbf{V}}_t^{T}$, and we have $2 {\rm{Re}} \big\{\boldsymbol{\theta}^{T}\widetilde{\mathbf{V}}_t\boldsymbol{\theta}^{*}\big\} = \boldsymbol{\theta}^{H}\overline{\mathbf{V}}_t\boldsymbol{\theta}$. Accordingly, a surrogate function for $\boldsymbol{\theta}^{H}\widetilde{\mathbf{V}}_t\boldsymbol{\theta} $ w.r.t. $\boldsymbol{\theta}$ can be derived as $\boldsymbol{\theta}^{H}\overline{\mathbf{V}}_t\boldsymbol{\theta}
        \leq
        \boldsymbol{\theta}^{H} \mathbf{\Lambda}_2 \boldsymbol{\theta}
        +
            2 {\rm{Re}} \big\{ 
                \boldsymbol{\theta}^{H}
                (\overline{\mathbf{V}}_t - \mathbf{\Lambda}_2)
                \boldsymbol{\theta}_t
            \big\}
        -
            \boldsymbol{\theta}_t^{H}
                (\overline{\mathbf{V}}_t - \mathbf{\Lambda}_2)
            \boldsymbol{\theta}_t$ 
where $\mathbf{\Lambda}_2 = \lambda_{max}(\overline{\mathbf{V}}_t)\mathbf{I}$. 
After these two majorization operations, an alternative upper bound optimization problem w.r.t. phase shift $\boldsymbol{\theta}$ is formulated as
    \begin{align}
    \begin{footnotesize}
            \begin{aligned} 
            \mathcal{P}_{1,MM}^{sub}: \ \min\limits_{\boldsymbol{\theta}}& \quad
            g_2(\boldsymbol{\theta};\boldsymbol{\theta}_{t})
            \notag
             \\ 
                s.t. \ 
                &{\rm{CR1}}: \ |\boldsymbol{\theta}(i)| = 1, 
            \end{aligned}
    \end{footnotesize}
    \end{align}
where $g_2(\boldsymbol{\theta};\boldsymbol{\theta}_{t})
            \triangleq
            2 {\rm{Re}} \big\{  
                \boldsymbol{\theta}^{H}\mathbf{f}_t   
            \big\}+ {\rm{const}}4$ and
$\mathbf{f}_t = (\overline{\mathbf{V}}_t - \lambda_{\max}(\mathbf{J}_2))\boldsymbol{\theta}_t + \widetilde{\mathbf{u}}_t$. Finally, the optimal solution for $\mathcal{P}_{1,MM}^{sub}$ admits the closed-form solution as
    \begin{align} \label{optimal_theta_MM}
    \begin{footnotesize}
    \begin{aligned}
        \boldsymbol{\theta}^{opt} =  e^{-j {\rm{arg}}(\mathbf{f}_t) }.
    \end{aligned}
    \end{footnotesize}
    \end{align}
The two-tier MM-based algorithm is summarized in \textbf{Algorithm 2}.

    \vspace{-5pt}  
    \begin{algorithm} \label{MM_algorithm}

    \renewcommand{\algorithmicrequire}{\textbf{Input:}}
    \renewcommand{\algorithmicensure}{\textbf{Output:}}
	\caption{Two-Tier MM-based Algorithm}
	\label{alg1}
	   \begin{algorithmic}[1]
	   \footnotesize
        		\REQUIRE The optimal $\boldsymbol{\theta}_n$, $\mathbf{p}_n$, $\boldsymbol{\eta}_n$ and $\boldsymbol{\chi}_n$ at the $n^{th}$ iteration of BCD; Maximum iteration $t_{\max}$.
                \ENSURE  $\boldsymbol{\theta}_{n+1}$
                \STATE Initialization: $t = 0$, $\boldsymbol{\theta}^{t} = \boldsymbol{\theta}_n$.
        		\REPEAT
        		\STATE Update $\boldsymbol{\theta}^{t+1}$ by (\ref{optimal_theta_MM});
                \STATE $t = t+1$;
        		\UNTIL  $f_q(\boldsymbol{\theta},\mathbf{p},\boldsymbol{\eta},\boldsymbol{\chi})$ in \ref{f_q} converges or $t > t_{\max}$;
                \RETURN $\boldsymbol{\theta}_{n+1} = \boldsymbol{\theta}^{t}$.
    	\end{algorithmic}  	
    \end{algorithm}

\vspace{-10pt}   
\subsection{Proposed Joint Design Algorithm}

    \vspace{-5pt}
    \begin{algorithm} \label{BCD_algorithm}

    \renewcommand{\algorithmicrequire}{\textbf{Input:}}
    \renewcommand{\algorithmicensure}{\textbf{Output:}}
	\caption{TTS Joint Transceiver Design Algorithm}
	\label{alg3}
	   \begin{algorithmic}[1]
	   \footnotesize
        		\REQUIRE Maximum iteration $n_{\max}$.
                \ENSURE  $\mathbf{p}^{Opt}$, $\boldsymbol{\theta}^{Opt}$
                \STATE Initialize: $\mathbf{p}_{0}$, $\boldsymbol{\theta}_{0}$ to feasible values.
                \STATE Initialize: $\boldsymbol{\eta}_{0}$, $\boldsymbol{\chi}_{0}$ by (\ref{Optimal_eta}), and (\ref{Optimal_chi}); $n = 0$.
        		\REPEAT
        		\STATE Update $\boldsymbol{\eta}_{n}$, $\boldsymbol{\chi}_{n}$, $\mathbf{\mathbf{p}}_{n}$;
                \STATE Update $\boldsymbol{\theta}_{n}$ with \textbf{Algorithm 1} or \textbf{Algorithm 2};
                \STATE $n = n+1$;
        		\UNTIL  $f_q(\boldsymbol{\theta},\mathbf{p},\boldsymbol{\eta},\boldsymbol{\chi})$ in \ref{f_q} converges or $n > n_{\max}$;
                \RETURN $\mathbf{p}^{Opt} = \mathbf{p}_{n}$, $\boldsymbol{\theta}^{Opt} = \boldsymbol{\theta}_{n}$.
    	\end{algorithmic}  	
    \end{algorithm}

    Now the proposed joint design algorithm can be finally summarized in \textbf{Algorithm 3} with the block update rule as follows
    \begin{align} \label{BCD_rule}
    \begin{footnotesize}
    \begin{aligned}
        \cdots \boldsymbol{\eta} \rightarrow \boldsymbol{\chi} \rightarrow \mathbf{p} \rightarrow \boldsymbol{\chi} \rightarrow \boldsymbol{\theta} \rightarrow \boldsymbol{\eta} \cdots
    \end{aligned}
    \end{footnotesize}
    \end{align}

\vspace{-10pt}  
\subsection{Convergence and Complexity Analysis}
    The convergence of the proposed joint design algorithm is illustrated in \textbf{Proposition 3}.
    
    \begin{proposition}
        Suppose the solution sequence generated by the proposed \textbf{Algorithm 3} is $\{\boldsymbol{\theta}_t, \mathbf{p}_t, \boldsymbol{\eta}_t, \boldsymbol{\chi}_t \}_{t = 0,1,\cdots,\infty}$. Then, the objective value $f_q(\boldsymbol{\theta}, \mathbf{p}, \boldsymbol{\eta}, \boldsymbol{\chi})$ keeps monotonically non-decreasing and finally converges to a stationary solution.
    \end{proposition}
    \begin{proof}
        See Appendix E. 
    \end{proof}

    In the sequel, we analyze the complexity of the proposed \textbf{Algorithm 3}. 
    To start with, it is clear that the values of \ref{Theorem2Signal}, \ref{Theorem2Noise}, \ref{Theorem2Leak}, \ref{Theorem2Interference} are fixed for given $\mathbf{A}$ and $\boldsymbol{\theta}$. They only need to be calculated once in each BCD iteration with a computation complexity of $\mathcal{O}(K^2 N^2)$. As a result, the complexity of updating dual variables, i.e., $\boldsymbol{\eta}$ and $\boldsymbol{\chi}$, as well as the user transmit powers $\mathbf{p}$ can be ignored since they have closed-form expressions which only depend on \ref{Theorem2Signal}, \ref{Theorem2Noise}, \ref{Theorem2Leak}, \ref{Theorem2Interference}, respectively. Next, for phase shift design, \textbf{Algorithm 1} and \textbf{Algorithm 2} are proposed. When \textbf{Algorithm 1} is applied for \textbf{Algorithm 3}, the complexity lies in the calculation of the Euclidean gradient in \ref{RGA_gradient} and the update of search step size. As such, the computation complexity is $\mathcal{O}(I_0(K^2N^2 + T_{RGA}(N^4 + T_{\rho} ))  )$, where $T_{RGA}$ is the iteration number for RGA updates and $T_{\rho}$ is the number of iterations for searching the step size $\rho_t$. On the other hand, the computation complexity of \textbf{Algorithm 2} depends on the calculation of the objective function $f_q$ in \ref{f_q} and the calculation of the largest eigenvalue of the corresponding matrix, which can be implemented by the low-cost power iteration method \cite{EVD}. To sum up, applying \textbf{Algorithm 2} for \textbf{Algorithm 3} has the complexity of $\mathcal{O}(I_0(K^2N^2 + T_{MM}(k_1 N^2 + k_2 (N+N^2)^2))  )$ where $I_0$, $T_{MM}$, $k_1$ and $k_2$ represent the number of iterations for BCD, \textbf{Algorithm 2}, and the iterations for calculating the largest eigenvalues of $\mathbf{J}_2$ and $\overline{\mathbf{V}}_t$, respectively.

\vspace{-10pt}    
\section{Simulation Results}
    In this section, numerical results are presented to validate the proposed algorithms' effectiveness and to demonstrate the superiority of the RDARS-aided massive MIMO systems.
    
\vspace{-10pt}   
\subsection{Simulation Setup}

    We consider the RDARS-aided uplink communication system where a RDARS is deployed to assist the communication between $4$ single-antenna users and a multi-antenna BS. As a common setting, a 3D-cartesian coordinate is adopted where the BS and the RDARS are located as (0m, 0m, 10m) and (0m, 0m, 20m), respectively. The four users are uniformly distributed in a circle centered at ($100$m, $-20$m, $1.5$m) with a radius of $10$m. 
    The log-distance large-scale path loss models from the 3GPP propagation environment \cite[Table B.1.2.1-1]{3GPP} are adopted as
    \begin{align}
        PL &= C_0 + 10 \alpha_{e}\log_{10}d, PL \in \{ \gamma_k, \alpha_k, \beta_k \}, \forall k \in \mathcal{K},
    \end{align}
    where $C_0$ is the path loss at reference distance, $d$ is the distance of the corresponding links, and $\alpha_{e}$ is the path loss exponent.  Unless specified otherwise, $C_0 = 30$, and the path-loss exponents of user-RDARS, RDARS-BS, and user-BS are set as $\{ \alpha_{e, UR}, \alpha_{e, RB}, \alpha_{e, UB} \} = \{ 2.3, 2.0, 3.5 \}$. Furthermore, the transmit and noise powers are set as $p_{\max} = 0$ dBm and $\sigma^2 = \sigma_B^2 = \sigma_R^2 =  -80$ dBm. The Rician factors are $\delta = 10$ and $\epsilon = 1$. Also, the AoA and AoD are generated randomly from [0,2$\pi$] \cite{Kangda_Zhi} and are fixed once generated. The weights $w_k, \forall k \in \mathcal{K}$ are chosen inversely proportional to the direct-link path loss, i.e., $\gamma_k, \forall k \in \mathcal{K} $, and then normalized by $\sum w_k = 1$. The number of symbols in each channel coherence time interval is $\tau_c = 196$ where $\tau = 8$ symbols are utilized for channel estimation\cite{Kangda_Zhi}.  We set $p_p = p_{\max}$ for simplicity. The space between the RDARS elements is set as $d_s = \lambda/ 2$. 
    Since the design of $\mathbf{A}$ is out of the scope of this paper, without loss of generality, the indicator matrix $\mathbf{A}$ is fixed and is configured by setting the top $a$ diagonal elements as 1, i.e., $\mathbf{A}^{H}\mathbf{A}(i, i) =1, \forall i \in \{ 1, \ldots, a\}$ which means the top $a$ elements on RDARS perform \textit{connected mode}.

    \vspace{-3pt}    
    \subsection{Weighted Sum Rate Maximization}

    In this subsection, we first verify the effectiveness and the convergence of the proposed TTS optimal design, i.e., \textbf{Algorithm 3}, denoted as (\textit{RDARS/ \ RIS Joint}). To this end, the following baselines are included for comparison: 
    \begin{itemize}
            \item  (\textit{RDARS/\ RIS RGA}), $p_i = p_{\max}, \forall i$, and phase shift $\boldsymbol{\theta}$ is optimized by \textbf{Algorithm 1}.
            \item  (\textit{RDARS/\ RIS MM}), $p_i = p_{\max}, \forall i$, and phase shift $\boldsymbol{\theta}$ is optimized by \textbf{Algorithm 2}.
            \item  (\textit{RDARS/\ RIS}), $\boldsymbol{\theta} = \mathbf{1}$, $p_i = p_{\max}, \forall i$.
            \item  (\textit{DAS Power}), $N = a$, and $p_i, \forall i$ is updated by \textbf{Algorithm 3}.
            \item  (\textit{W.O. RDARS Power}), $N = 0$, $a = 0$, and $p_i, \forall i$ is updated by \textbf{Algorithm 3}.
    \end{itemize}


    
    To begin with, Fig. 3 shows the convergence behavior of the proposed algorithms and other baselines. It is clear that with the optimal joint design, RDARS-aided system can achieve higher sum rate than the DAS and RIS-aided systems. This is mainly due to the joint exploration of the distribution and reflection gains in RDARS. Moreover, the proposed ``\textit{RDARS Joint}" outperforms both ``\textit{RDARS MM}'' and ``\textit{RDARS RGA}'' in terms of the weighted sum rate through the joint design of both phase shifts and user transmit powers. Besides, the ``Two-tier MM-based algorithm'' and the ``RGA-based algorithm'' for phase shift design in RDARS-aided or RIS-aided systems can achieve almost the same converged value. In contrast, the ``RGA-based algorithm'' has a faster convergence speed than the  ``Two-tier MM-based algorithm'' in RIS-aided systems. In addition, it is noted that the proposed algorithms converge monotonically within $10$ iterations, which further corroborates their fast convergence speed. Based on the effectiveness and the convergence of the proposed \textbf{Algorithm 3} illustrated in Fig. 3, we only consider ``\textit{RDARS Joint}", ``\textit{RIS Joint}", ``\textit{DAS Power}'' and ``\textit{W.O. RDARS Power}'' in the following simulations for simplicity.

    \begin{figure} [t]  
        \small
        \setlength{\belowcaptionskip}{-3mm}  
        \centering
        \includegraphics[width=0.80\columnwidth]{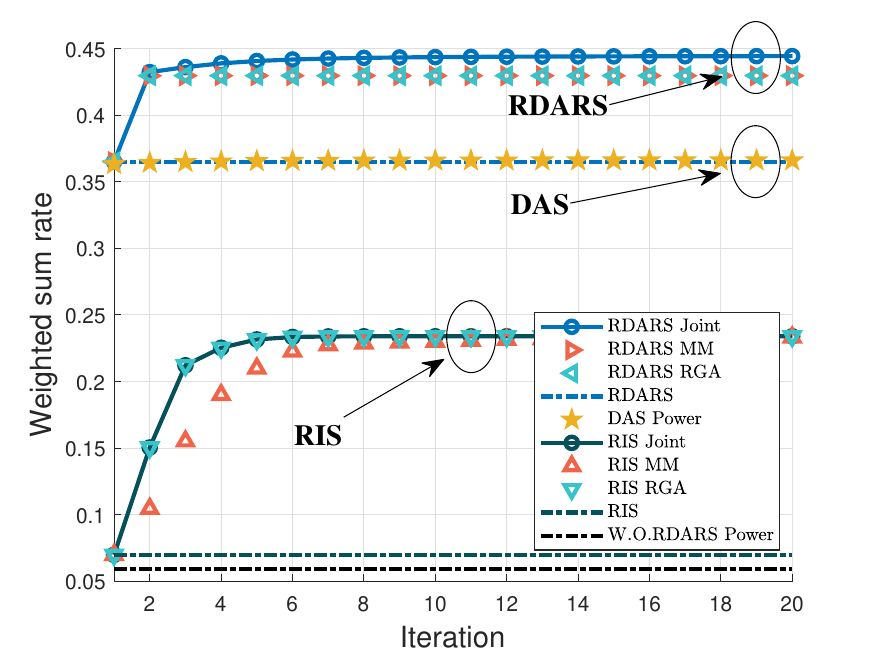} 
        \caption{ 
        Convergence behavior of the proposed algorithms and different baselines. $L = 128$, $N = 32$, $a = 2$ and $p_{\max} = 0$ dBm.}     
    \end{figure}
    
    \begin{figure} [t]  
        \small
        \setlength{\belowcaptionskip}{-3mm}  
        \centering
        \includegraphics[width=0.80\columnwidth]{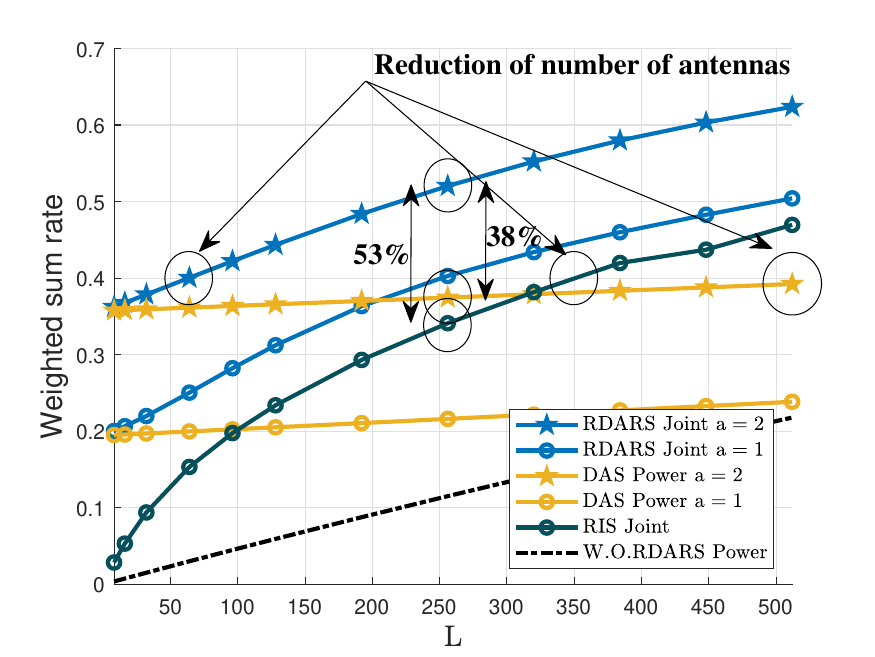} 
        \caption{ 
        The weighted sum rate [bps/Hz] versus the number of BS antennas $L$ when $N = 32$ and $p_{\max} = 0$ dBm.}     
    \end{figure}

     In Fig. 4, the weighted sum rate is provided against the number of BS antennas $L$. It is clear that with the help of RDARS, we can achieve the same rate as its counterparts but with a much-reduced number of active antennas/RF chains. For example, the rate obtained by RIS with a $350$-antennas BS can be obtained by a $64$-antennas BS in an RDARS-aided massive MIMO system with only $2$ elements acting \textit{connected mode} on a $32$-elements RDARS. Besides, RDARS attains $53$\% and $38$\% rate increments compared to RIS-aided and DAS systems with a $256$-antennas BS. Since the cost and the energy consumption of elements acting \textit{reflection mode} are much lower than that of active antennas applied on conventional massive MIMO systems, incorporating RDARS into massive MIMO is a promising and cost-effective solution for future wireless communication.


    \vspace{-0.2cm}
    \begin{figure} [!htb]  
        \small
        \setlength{\belowcaptionskip}{-3mm}  
        \centering
        \includegraphics[width=0.80\columnwidth]{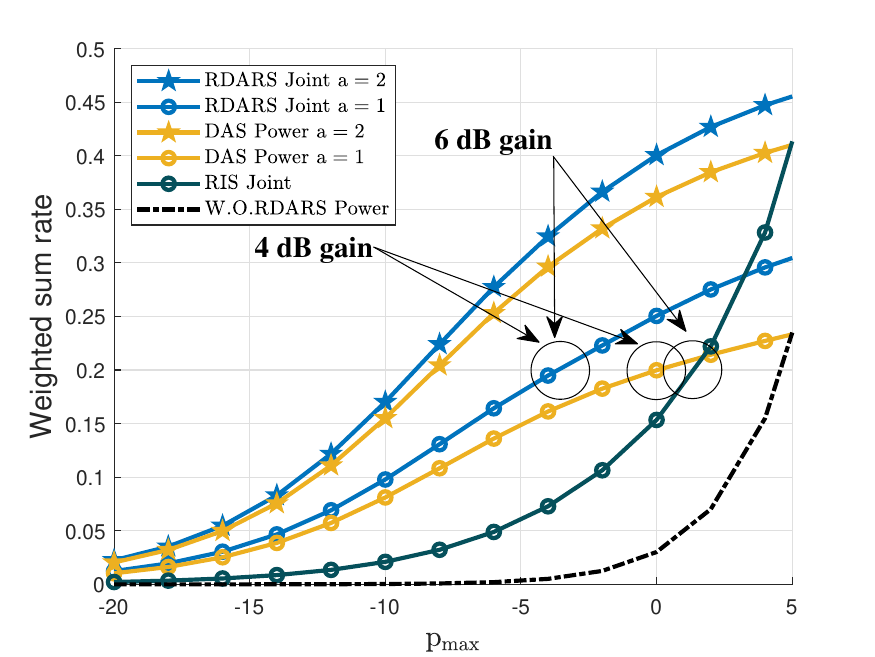} 
        \caption{ 
        The weighted sum rate [bps/Hz] versus the maximum transmit power $p_{\max}$ [dBm] when $L = 64$ and $N = 32$.}     
    \end{figure}

     Fig. 5 illustrates the weighted sum rate versus the maximum transmit power $p_{\max}$. First, the RDARS-aided system outperforms its counterparts given a wide range of $p_{\max}$. For example, to achieve a $0.2$ bps/Hz weighted sum rate, the required transmit powers for DAS and RIS-aided systems are about $0$ dBm and $2$ dBm, respectively. In contrast, the RDARS-aided system with only single element acting \textit{connected mode} only requires $-4$ dBm. This demonstrates $4$ dB and $6$ dB power-savings provided by the RDARS-aided system compared to the DAS and the RIS-aided systems, respectively. Besides, the results also indicate that using \textit{connected mode} on RDARS is favorable in a low transmit power regime, when comparing the rates of RDARS-aided and RIS-aided systems. But with the increase of the transmit power, the superiority of RDARS with connected-mode elements over RIS is decreasing. The rationale behind this phenomenon is that the rate will be interference-limited in a high transmit power regime and the connected-mode elements on RDARS may not always contribute positively due to the high interference received.  Nevertheless, with the additional benefits from ``reconfigurability" of RDARS, which allows for dynamically configuring different modes for each element \cite{RDARS}, RDARS has more freedom and flexibility to adapt to the communication environments and then boost the performance further. The reconfigurability is achieved through the design of the indicator matrix $\mathbf{A}$ and  will be exploited deeply in our future work. Finally, Fig. 6 shows the weighted sum rate versus the number of RDARS elements $N$. It is noted that a remarkable performance gain is attained by RDARS when compared to its counterparts.

    \vspace{-0.3cm}
    \begin{figure} [!htb]  
        \small
        \setlength{\belowcaptionskip}{-3mm}  
        \centering
        \includegraphics[width=0.80\columnwidth]{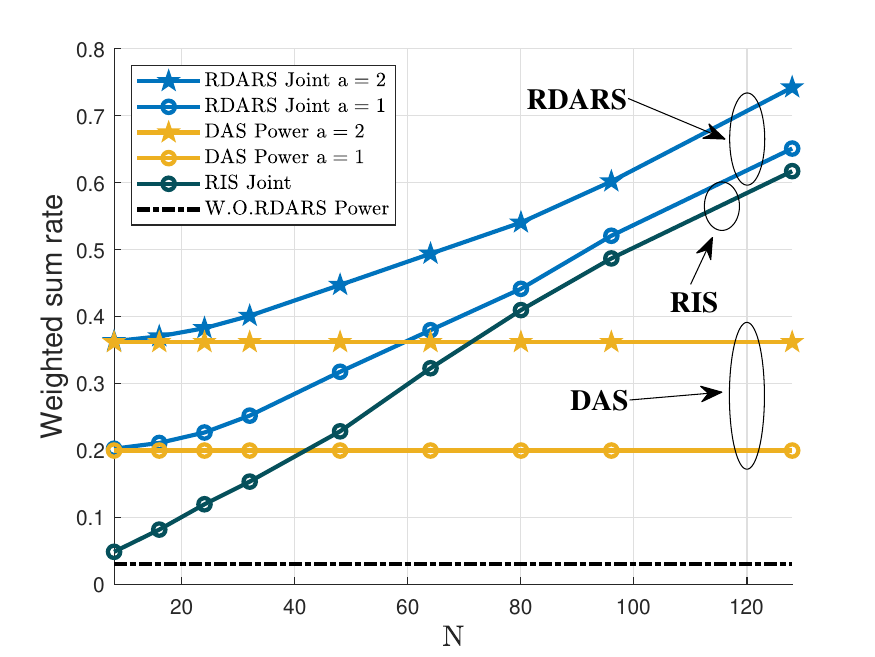} 
        \caption{ 
        The weighted sum rate [bps/Hz] versus the number of RDARS elements $N$ when $L = 64$, $a = 2$ and $p_{\max} = 0$ dBm.}     
    \end{figure}
    \vspace{-0.05cm}

\vspace{-3pt}   
\section{Conclusion}
    This paper investigated a RDARS-aided massive MIMO uplink communication scenario with the TTS transceiver design. The achievable rate under MRC and imperfect CSI was analyzed in closed-form expression, which then enabled optimal TTS transceiver design with reduced overhead and power consumption. An efficient algorithm with guaranteed convergence to the stationary solution was proposed to jointly optimize the users' transmit powers and RDARS phase shifts to maximize the ergodic weighted sum rate. Simulation results validated the effectiveness of the proposed algorithm and demonstrated the superiority of deploying RDARS in MIMO systems to provide substantial rate improvement with a significantly reduced total number of active antennas/RF chains and lower transmit power when compared to its counterparts.

    This paper is an initial research on RDARS-aided multi-user communications. It is believed that RDARS can provide not only distribution gain and reflection gain but also selection gain through the mode selection for each element. The analysis of the selection gain and the optimal design of the element working modes are to be further exploited. The potential of RDARS under other scenarios, e.g., downlink transmission, localization, and integrated sensing and communications, is also worth investigating in the future.

\vspace{-10pt}   
\begin{appendices}
    \section{Proof of \textbf{Lemma 1}}
    Using the independence of $\mathbf{N}$ and $\mathbf{Q}$, the first order moment of $\mathbf{y}_{k,p}$ in \eqref{y_k} can be easily obtained as
    \begin{align}
    \begin{footnotesize}
    \begin{aligned}
        \mathbb{E}[\mathbf{y}_{k,p}] 
        &=
        \mathbb{E}[\mathbf{q}_k] 
       + \frac{1}{\sqrt{\tau p_p}}\mathbb{E}[\mathbf{N}]\mathbf{s}_{k,p} 
       = \mathbb{E}[\mathbf{q}_k]
       =
            \begin{bmatrix}
                \sqrt{c_k \delta \epsilon_k} \overline{\mathbf{H}}\mathbf{B} \overline{\mathbf{h}}_k \\
                \sqrt{d_k \epsilon_k} \mathbf{A}\overline{\mathbf{h}}_k
            \end{bmatrix}.
    \end{aligned}
    \end{footnotesize}
    \end{align}
    Also, we have $\mathbb{C}[\mathbf{y}_{k,p}, \mathbf{y}_{k,p}] = \mathbb{E}[ ( \mathbf{y}_{k,p} - \mathbb{E}[\mathbf{y}_{k,p}]) (\mathbf{y}_{k,p} - \mathbb{E}[\mathbf{y}_{k,p}])^{H}]$. Recalling the definitions of $\mathbf{q}_k$, $\mathbf{A}$ and $\mathbf{B}$ in Section II, and leveraging identity $\mathbf{B}\mathbf{A}^{H} = \mathbf{0}$, it is shown that the covariance matrix for $\mathbf{q}_k$ is a block diagonal matrix. Thus, we have
    \begin{align}
    \begin{footnotesize}
    \begin{aligned}
        \mathbb{C}[\mathbf{q}_k, \mathbf{y}_{k,p}]
       &=
       \mathbb{C}[\mathbf{q}_k, \mathbf{q}_k]
       =
       \begin{bmatrix}
            \mathbb{C}[\mathbf{h}_{B,k}, \mathbf{h}_{B,k}] & \mathbf{0} & \\
            \mathbf{0} &
            \mathbb{C}[\mathbf{h}_{R,k}, \mathbf{h}_{R,k}] 
        \end{bmatrix},
    \end{aligned}
    \end{footnotesize}
    \end{align}
    where $\mathbb{C}[\mathbf{h}_{B,k}, \mathbf{h}_{B,k}] = (N-a) c_k \delta \mathbf{a}_L \mathbf{a}_L^{H} + ( (N-a)c_k(\epsilon_k + 1) + \gamma_k ) \mathbf{I}_L$, and $\mathbb{C}[\mathbf{h}_{R,k},\mathbf{h}_{R,k}^{H}] =\mathbb{E}[ (\mathbf{h}_{R,k} - \mathbb{E}[\mathbf{h}_{R,k} ]) (\mathbf{h}_{R,k} - \mathbb{E}[\mathbf{h}_{R,k} ])^{H} ] = d_k \mathbf{I}_a$. The proof is thus completed. 

    \vspace{-10pt}  
    \section{Proof of \textbf{Theorem 1}}
        The LMMSE estimate of \textit{equivalent channel} $\mathbf{q}_k$ based on the received pilot signal $\mathbf{y}_{k,p}$ is obtained as \cite{SIG-093}
        \begin{align}
        \begin{footnotesize}
        \begin{aligned}
            \hat{\mathbf{q}_k} = 
            \mathbb{E}[\mathbf{q}_k] +
            \mathbb{C}[\mathbf{q}_k,\mathbf{y}_{k,p}] \mathbb{C}[\mathbf{y}_{k,p},\mathbf{y}_{k,p}]^{-1} (\mathbf{y}_{k,p} - \mathbb{E}[\mathbf{y}_{k,p}]).
        \end{aligned}
        \end{footnotesize}
        \end{align}
    Using the lemma of the inverse of block matrix and Woodbury matrix identity \cite{MatrixCookbook}, we obtain
        \begin{align}
        \begin{footnotesize}
        \begin{aligned}
            \mathbb{C}[\mathbf{q}_k, \mathbf{y}_{k,p}]\mathbb{C}[\mathbf{y}_{k,p}, \mathbf{y}_{k,p}]^{-1} 
            &=
            \begin{bmatrix}
                \mathbf{E}_{k} &  \mathbf{0}_{L \times a} \\
                \mathbf{0}_{a \times L} & a_{k5} \mathbf{I}_a
            \end{bmatrix}.
        \end{aligned}
        \end{footnotesize}    
        \end{align}
    By using $\mathbb{E}[\mathbf{q}_k] = \mathbb{E}[\mathbf{y}_{k,p}]$, the LMMSE channel estimate is obtained as
    \begin{align}
    \begin{footnotesize}
    \begin{aligned}
        \hat{\mathbf{q}}_k 
        \triangleq
        \mathbf{C}_k \mathbf{y}_{k,p} + \mathbf{D}_k,
    \end{aligned}
    \end{footnotesize}
    \end{align}
    where $\mathbf{C}_k$ and $\mathbf{D}_k$ are defined as
    \begin{align}
    \begin{footnotesize}
    \begin{aligned}
        \mathbf{C}_k 
        &=
            \begin{bmatrix}
                \mathbf{E}_k & \mathbf{0}_{L \times a} \\
                \mathbf{0}_{a \times L} & a_{k5}\mathbf{I}_a
            \end{bmatrix}
        ,
        \mathbf{D}_k 
        =
            \begin{bmatrix}
                (\mathbf{I}_L - \mathbf{E}_k)\sqrt{c_k \delta \epsilon_k} \overline{\mathbf{H}}\mathbf{B}\overline{\mathbf{h}}_k \\
                (1- a_{k5}) \sqrt{d_k\epsilon_k} \mathbf{A}\overline{\mathbf{h}}_k
            \end{bmatrix}.
    \end{aligned}
    \end{footnotesize}
    \end{align} 
    
    As such, we obtain the channel estimate using the LMMSE estimator in (\ref{EstimatedChannel})
    where the equality $\mathbf{E}_k \overline{\mathbf{H}} = (a_{k3}L + a_{k4})\overline{\mathbf{H}}$ is utilized. This completes the proof.

    \vspace{-10pt}  
    \section{Proof of \textbf{Theorem 2}}
    Here we derive $E_{k}^{signal}$, $E_{k}^{Noise}$, $E_{k}^{Leak}$ and $I_{k,i}$ in \ref{Theorem2Signal}, \ref{Theorem2Noise}, \ref{Theorem2Leak} and \ref{Theorem2Interference}, respectively, as follows.
    \paragraph{Signal and Noise Terms}

    Leveraging the unbiasness and orthogonality of the LMMSE estimator, i.e., $\mathbb{E}[\mathbf{e}_k\mathbf{y}_k^{H}] = \mathbf{0}$ and $\mathbb{E}[\mathbf{e}_k] = \mathbf{0}$, we have
    $\mathbb{E}[\hat{\mathbf{q}}_k^{H} \mathbf{q}_k] =\mathbb{E}[\hat{\mathbf{q}}_k^{H}(\mathbf{e}_{k} + \hat{\mathbf{q}}_k) ]$. Accordingly, we have $\sqrt{E_{k}^{signal}(\mathbf{A},\boldsymbol{\Theta})} = \mathbb{E}[\hat{\mathbf{q}}_k^{H} \mathbf{q}_k]$, and thus the signal and noise terms can be derived by identifying the non-zero terms following similar procedure as in \cite{Kangda_Zhi}.
    
    \paragraph{Interference Term}

    \newcounter{Interference_Appendix}
        \begin{figure*}[!htb] 
        \scriptsize
        \setcounter{Interference_Appendix}{\value{equation}}
        \setcounter{equation}{43}
        \begin{align} 
            &I_{ki} (\mathbf{A},\mathbf{\Theta})
            =  
                \underbrace{
                    \mathbb{E}\bigg[ \bigg|
                        \hat{\underline{\mathbf{h}}}_{B,k}^{H}
                        \underline{\mathbf{h}}_{B,i}
                    \bigg|^2 \bigg]
                }_{\widetilde{1}}
            +
                \underbrace{
                    \mathbb{E}\bigg[ \bigg|
                         \hat{\underline{\mathbf{h}}}_{B,k}^{H}
                         \mathbf{d}_{i}
                    \bigg|^2 \bigg]
                }_{\widetilde{2}}
            +   
                \underbrace{
                    \mathbb{E}\bigg[ \bigg|
                         \mathbf{d}_{k}^{H}
                         \mathbf{E}_{k}^{H}
                         \underline{\mathbf{h}}_{B,i}
                    \bigg|^2 \bigg]
                }_{\widetilde{3}}
            +
                \underbrace{
                    \mathbb{E}\bigg[ \bigg|
                         \mathbf{d}_{k}^{H}
                         \mathbf{E}_{k}^{H}
                         \mathbf{d}_{i}
                    \bigg|^2 \bigg]
                }_{\widetilde{4}}
            +
                \underbrace{
                    \frac{1}{\tau p_p}
                    \mathbb{E}\bigg[ \bigg|
                         \mathbf{s}_{k}^{H}
                         \mathbf{N}_{B}^{H}
                         \mathbf{E}_{k}^{H}
                         \underline{\mathbf{h}}_{B,i}
                    \bigg|^2 \bigg]
                }_{\widetilde{5}}
            +   
                \underbrace{
                    \frac{1}{\tau p_p}
                    \mathbb{E}\bigg[ \bigg|
                         \mathbf{s}_{k}^{H}
                         \mathbf{N}_{B}^{H}
                         \mathbf{E}_{k}^{H}
                         \mathbf{d}_{i}
                    \bigg|^2 \bigg]
                }_{\widetilde{6}}
            \notag
            \\
            &+
                \underbrace{
                    d_{k}\epsilon_{k}
                    d_{i}\epsilon_{i}
                    \mathbb{E}\bigg[ \bigg|
                         \overline{\mathbf{h}}_{k}^{H}\mathbf{A}^{H}
                         \mathbf{A} \overline{\mathbf{h}}_{i}
                    \bigg|^2 \bigg]
                }_{\widetilde{7}}
            +
                \underbrace{
                    d_{k}\epsilon_{k}d_{i}
                    \mathbb{E}\bigg[ \bigg|
                        \overline{\mathbf{h}}_{k}^{H}\mathbf{A}^{H}
                        \mathbf{A} \widetilde{\mathbf{h}}_{i}
                    \bigg|^2 \bigg]
                }_{\widetilde{8}}
            +
                \underbrace{
                     d_{k} d_{i}\epsilon_{i} e_{k4}^{2} 
                    \mathbb{E}\bigg[ \bigg|
                        \widetilde{\mathbf{h}}_{k}^{H} \mathbf{A}^{H} 
                        \mathbf{A} \overline{\mathbf{h}}_{i}
                    \bigg|^2 \bigg]
                }_{\widetilde{9}}
            +
                \underbrace{
                    d_{k} d_{i} e_{k4}^2
                    \mathbb{E}\bigg[ \bigg|
                        \widetilde{\mathbf{h}}_{k}^{H} \mathbf{A}^{H} 
                        \mathbf{A} \widetilde{\mathbf{h}}_{i}
                    \bigg|^2 \bigg]
                }_{\widetilde{10}}
            \notag
            \\
            &+
                \underbrace{
                    \frac{1}{\tau p_p} d_{i}\epsilon_{i} e_{k4}^2 
                    \mathbb{E}\bigg[ \bigg|
                        \mathbf{s}_{k}^{H} \mathbf{N}_{R}^{H} 
                          \mathbf{A} \overline{\mathbf{h}}_{i}
                    \bigg|^2 \bigg]
                }_{\widetilde{11}}
            +
                \underbrace{
                    \frac{1}{\tau p_p} d_{i} e_{k4}^2 
                    \mathbb{E}\bigg[ \bigg|
                        \mathbf{s}_{k}^{H} \mathbf{N}_{R}^{H}
                        \mathbf{A} \widetilde{\mathbf{h}}_{i}
                    \bigg|^2 \bigg]
                }_{\widetilde{12}}
            +
                 \underbrace{
                    2 {\rm{Re}} \bigg\{
                    \sqrt{d_{k}\epsilon_{k}}
                    \sqrt{d_{i}\epsilon_{i}}
                    \mathbb{E}\bigg[ 
                        \hat{\underline{\mathbf{h}}}_{B,k}^{H}
                        \underline{\mathbf{h}}_{B,i}
                        \overline{\mathbf{h}}_{i}^{H}
                        \mathbf{A}^{H}\mathbf{A}
                        \overline{\mathbf{h}}_{k}
                     \bigg]
                    \bigg\}
                }_{\widetilde{13}}
                \tag*{\normalsize (44)}
                \label{Interference_Appendix}
        \end{align}
        \setcounter{equation}{\value{Interference_Appendix}}
        \hrulefill
        \end{figure*}
        \addtocounter{equation}{1}

    Recalling that $I_{ki} = \sum\nolimits_{i\neq k}^{K} p_i \mathbb{E}\big[|\hat{\mathbf{q}}_k^{H} \mathbf{q}_i|^2\big]$, we can rewrite the expression as \ref{Interference_Appendix}, where the terms $\widetilde{1}$ to $\widetilde{6}$ can be obtained following similar approach as in \cite{Kangda_Zhi}. As such, we focus on the derivations of $\widetilde{7}$ to $\widetilde{13}$. Specifically, we have 
    $\widetilde{7} = 
                        d_{k} d_{i} \epsilon_{k}\epsilon_{i} |g_{k,i} (\mathbf{A}) |^2$,
    $\widetilde{8}
                    =
                        a d_{k}\epsilon_{k}d_{i}$,
    $\widetilde{9}
                    =
                        a d_{k} d_{i}\epsilon_{i} e_{k4}^{2}$,
    $\widetilde{10}
                    =
                        a d_{k} d_{i} e_{k4}^2$,
    $\widetilde{11}
                =
                    a \frac{\sigma_{R}^2}{\tau \rho} d_{i}\epsilon_{i} e_{k4}^2$,
    $\widetilde{12}
                =
                    a \frac{\sigma_{R}^2}{\tau \rho} d_{i} e_{k4}^2$,
    $\widetilde{13}
                =
                    2 {\rm{Re}} \{
                    L \sqrt{c_k c_i d_k d_i} \delta \epsilon_k \epsilon_i 
                    f_{k}^{H}(\mathbf{B}) f_{i}(\mathbf{B}) g_{i,k}(\mathbf{A})
                    +
                    L \sqrt{c_k c_i d_{k} d_{i}} \epsilon_k \epsilon_i
                      e_{k1}
                     (\overline{\mathbf{h}}_k^{H}\overline{\mathbf{h}}_i - g_{k,i}(\mathbf{A}) ) g_{i,k}(\mathbf{A})
                    \}$.

    \paragraph{Signal Leakage Term}
    
    The signal leakage term is calculated as $E_{k}^{leak}(\mathbf{A},\mathbf{\Theta})
     = \mathbb{E}[ |\hat{\mathbf{q}}_{k}^{H} \mathbf{q}_k|^2] - |\mathbb{E}[\hat{\mathbf{q}}_{k}^{H} \mathbf{q}_k]|^2$ where we already have $\mathbb{E}[\hat{\mathbf{q}}_{k}^{H} \mathbf{q}_k]$. As such, the remaining calculation is on the term $\mathbb{E}[ |\hat{\mathbf{q}}_{k}^{H} \mathbf{q}_k|^2]$ as in  \ref{Leakage_Appendix}. Here the terms $\widetilde{1}$ to $\widetilde{7}$ can also be obtained following similar approach as in \cite{Kangda_Zhi}. Therefore, we focus on the derivations of the terms $\widetilde{8}$ to $\widetilde{18}$ as follows. Specifically, we have 
     $\widetilde{8}
        =
            a d_k^2 \epsilon_k $,
    $\widetilde{9}
        =
            a d_k^2 \epsilon_k e_{k4}^2$,
    $\widetilde{10}
        =
            a^2 d_k^2 e_{k4}^2 + a d_k^2 e_{k4}^2$,
    $\widetilde{11}
        =
            a \frac{\sigma_R^2}{\tau \rho} d_k \epsilon_k e_{k4}^2 $,
    $\widetilde{12}
        =
            a\frac{\sigma_R^2}{\tau \rho} d_k e_{k4}^2 $.
    The other terms $\widetilde{13}$-$\widetilde{18}$ are calculated by first identifying the non-zero expectation terms, and then calculated as follows: 
    $\widetilde{13}
        =
            2 L^2 |f_{k}(\mathbf{A},\mathbf{\Theta}) |^2 c_k \delta \epsilon_k \gamma_k e_{k1}
            +
            2 L^2 (N-a) |f_{k}(\mathbf{A},\mathbf{\Theta}) |^2 c_k \delta \gamma_k e_{k1} e_{k2}
        +
            2 L^2 (N-a) |f_{k}(\mathbf{A},\mathbf{\Theta}) |^2 c_k (\epsilon_k + 1) \gamma_k e_{k1}^2$,
    $\widetilde{14}
        =
            2 L a |f_{k}(\mathbf{A},\mathbf{\Theta}) |^2 c_k d_k \delta \epsilon_k^2
            +
            2 L a (N-a) c_k d_k \delta \epsilon_k e_{k2}
            +
            2 L a (N-a) c_k d_k \epsilon_k^2 e_{k1}
            +
            2 L a (N-a) c_k d_k \epsilon_k e_{k1}$,
    $\widetilde{15} 
        =
            2 L a |f_{k}(\mathbf{A},\mathbf{\Theta}) |^2 c_k d_k \delta \epsilon_k e_{k4}
            +
            2 L a (N-a) c_k d_k \delta e_{k2} e_{k4}
        +
            2 L a (N-a) c_k d_k \epsilon_k e_{k1} e_{k4}
            +
            2 L a (N-a) c_k d_k e_{k1} e_{k4}$,
    $\widetilde{16}
        =
            2 L a d_k \gamma_k \epsilon_k e_{k1}$,
    $\widetilde{17}
        =
            2 L a d_k \gamma_k e_{k1} e_{k4}$,
    $\widetilde{18}
        =
           2 a^2 d_k^2 \epsilon_k e_{k4}$.
    Therefore, we have completed the calculation of the required statistics for the relevant terms. This thus completes the proof.

        \newcounter{Leakage_Appendix}
        \begin{figure*}[!htb]
        \scriptsize
        \setcounter{Leakage_Appendix}{\value{equation}}
        \setcounter{equation}{44}
        \begin{align} 
        &\mathbb{E} [|\hat{\mathbf{q}}_k^{H} \mathbf{q}_k|^2]
        =
                \underbrace{
                    \mathbb{E} \bigg[ 
                        \bigg| 
                            \hat{\underline{\mathbf{h}}}_{B,k}^{H} \underline{\mathbf{h}}_{B,k}
                        \bigg|^2
                    \bigg]
                }_{\widetilde{1}}
            +
                \underbrace{
                    \mathbb{E} \bigg[ 
                        \bigg| 
                            \hat{\underline{\mathbf{h}}}_{B,k}^{H} 
                            \mathbf{d}_k
                        \bigg|^2
                    \bigg]
                }_{\widetilde{2}}
            +   
                \underbrace{
                    \mathbb{E} \bigg[ 
                        \bigg| 
                            \mathbf{d}_k^{H}
                            \mathbf{E}_k^{H}
                            \underline{\mathbf{h}}_{B,k} 
                        \bigg|^2
                    \bigg]
                }_{\widetilde{3}}
            +
                \underbrace{
                    \mathbb{E} \bigg[ 
                        \bigg| 
                            \mathbf{d}_k^{H}
                            \mathbf{E}_k^{H}
                            \mathbf{d}_k
                        \bigg|^2
                    \bigg]
                }_{\widetilde{4}}
            +
                \underbrace{
                    \frac{1}{\tau p_p}
                    \mathbb{E} \bigg[ 
                        \bigg| 
                            \mathbf{s}_k^{H} \mathbf{N}_{B}^{H}
                            \mathbf{E}_k^{H} \underline{\mathbf{h}}_{B,k}
                        \bigg|^2
                    \bigg]
                }_{\widetilde{5}}
        +
                \underbrace{
                    \frac{1}{\tau p_p}
                    \mathbb{E} \bigg[ 
                        \bigg| 
                            \mathbf{s}_{k}^{H} \mathbf{N}_{B}^{H}
                            \mathbf{E}_{k}^{H} \mathbf{d}_k
                        \bigg|^2
                    \bigg]
                }_{\widetilde{6}}
        \notag
        \\
            &+
                \underbrace{
                    d_k^2 \epsilon_k^2
                    \mathbb{E} \bigg[ 
                        \bigg| 
                            \overline{\mathbf{h}}_k^{H} 
                            \mathbf{A}^{H} \mathbf{A}
                            \overline{\mathbf{h}}_k
                        \bigg|^2
                    \bigg]
                }_{\widetilde{7}}
            +
                \underbrace{
                    d_k^2 \epsilon_k 
                    \mathbb{E} \bigg[ 
                        \bigg| 
                            \overline{\mathbf{h}}_k^{H} 
                            \mathbf{A}^{H} \mathbf{A}
                            \widetilde{\mathbf{h}}_k
                        \bigg|^2
                    \bigg]
                }_{\widetilde{8}}
            +
                \underbrace{
                    d_k^2 \epsilon_k e_{k4}^2 
                    \mathbb{E} \bigg[ 
                        \bigg| 
                            \widetilde{\mathbf{h}}_k^{H} 
                            \mathbf{A}^{H} \mathbf{A}
                            \overline{\mathbf{h}}_k
                        \bigg|^2
                    \bigg]
                }_{\widetilde{9}}
        +
                \underbrace{
                    d_k^2 e_{k4}^2 
                    \mathbb{E} \bigg[ 
                        \bigg| 
                            \widetilde{\mathbf{h}}_k^{H} 
                            \mathbf{A}^{H} \mathbf{A}
                            \widetilde{\mathbf{h}}_k
                        \bigg|^2
                    \bigg]
                }_{\widetilde{10}}
        +
                \underbrace{
                    \frac{1}{\tau p_p} d_k \epsilon_k e_{k4}^2 
                    \mathbb{E} \bigg[ 
                        \bigg| 
                                \mathbf{s}_k^{H}
                                \mathbf{N}_R^{H}
                                \mathbf{A}
                                \overline{\mathbf{h}}_k
                        \bigg|^2
                    \bigg]
                }_{\widetilde{11}}
    \notag
    \\
        &+
                \underbrace{
                    \frac{1}{\tau p_p} d_k e_{k4}^2 
                    \mathbb{E} \bigg[ 
                        \bigg| 
                                \mathbf{s}_k^{H}
                                \mathbf{N}_R^{H}
                                \mathbf{A}
                                \widetilde{\mathbf{h}}_k
                        \bigg|^2
                    \bigg]
                }_{\widetilde{12}}
        +  
                \underbrace{
                    2 {\rm{Re}} \bigg\{
                        \mathbb{E} \bigg[ 
                                \hat{\underline{\mathbf{h}}}_{B,k}^{H} 
                                \underline{\mathbf{h}}_{B,k}
                                \mathbf{d}_k^{H}
                                \mathbf{E}_k^{H}
                                \mathbf{d}_k
                        \bigg]
                    \bigg\}
                }_{\widetilde{13}}
            +
                \underbrace{
                    2 {\rm{Re}} \bigg\{ 
                        d_k \epsilon_k
                        \mathbb{E} \bigg[ 
                                \hat{\underline{\mathbf{h}}}_{B,k}^{H} 
                                \underline{\mathbf{h}}_{B,k}
                                \overline{\mathbf{h}}_k^{H}
                                \mathbf{A}^{H} \mathbf{A}
                                \overline{\mathbf{h}}_k
                        \bigg]
                    \bigg\}
                }_{\widetilde{14}}
            +
                \underbrace{
                    2 {\rm{Re}} \bigg\{ 
                        d_k e_{k4}
                        \mathbb{E} \bigg[ 
                                \hat{\underline{\mathbf{h}}}_{B,k}^{H} 
                                \underline{\mathbf{h}}_{B,k}
                                \widetilde{\mathbf{h}}_k^{H}
                                \mathbf{A}^{H} \mathbf{A}
                                \widetilde{\mathbf{h}}_k
                        \bigg]
                    \bigg\}
                }_{\widetilde{15}}
        \notag
            \\
        &+
                \underbrace{
                    2 {\rm{Re}} \bigg\{ 
                        d_k \epsilon_k
                        \mathbb{E} \bigg[ 
                                \mathbf{d}_k^{H}
                                \mathbf{E}_k^{H}
                                \mathbf{d}_k
                                \overline{\mathbf{h}}_k^{H}
                                \mathbf{A}^{H} \mathbf{A}
                                \overline{\mathbf{h}}_k
                        \bigg]
                    \bigg\}
                }_{\widetilde{16}}
            +
                \underbrace{
                    2 {\rm{Re}} \bigg\{ 
                        d_k e_{k4}
                        \mathbb{E} \bigg[ 
                                \mathbf{d}_k^{H}
                                \mathbf{E}_k^{H}
                                \mathbf{d}_k
                                \widetilde{\mathbf{h}}_k^{H}
                                \mathbf{A}^{H} \mathbf{A}
                                \widetilde{\mathbf{h}}_k
                        \bigg]
                    \bigg\}
                }_{\widetilde{17}}
        +   
                \underbrace{
                    2 {\rm{Re}} \bigg\{ 
                        d_k^2 \epsilon_k e_{k4}
                        \mathbb{E} \bigg[ 
                                \overline{\mathbf{h}}_k^{H}
                                \mathbf{A}^{H}\mathbf{A}
                                \overline{\mathbf{h}}_k
                                \widetilde{\mathbf{h}}_k^{H}
                                \mathbf{A}^{H} \mathbf{A}
                                \widetilde{\mathbf{h}}_k
                        \bigg]
                    \bigg\}
                }_{\widetilde{18}}
            \tag*{\normalsize (45)}
            \label{Leakage_Appendix}
    \end{align}
        \setcounter{equation}{\value{Leakage_Appendix}}
        \hrulefill
        \end{figure*}
        \addtocounter{equation}{1}

    \vspace{-10pt}  
    \section{Proof of \textbf{Lemma 3}}

    To provide a simple expression for the gradient of $f_q$ w.r.t. $\mathbf{\theta}_{\widetilde{\mathcal{A}}}$, the objective function is first converted into the following compact form as
    \begin{align}
    \begin{footnotesize}
    \begin{aligned}
    f_q(\boldsymbol{\theta}_{\widetilde{\mathcal{A}}},\mathbf{p},\boldsymbol{\eta},\boldsymbol{\chi})  
        &=
            (\boldsymbol{\theta}_{\widetilde{\mathcal{A}}}^{*} \otimes \boldsymbol{\theta}_{\widetilde{\mathcal{A}}})^{H}
            \widetilde{\mathbf{J}}
            (\boldsymbol{\theta}_{\widetilde{\mathcal{A}}}^{*} \otimes \boldsymbol{\theta}_{\widetilde{\mathcal{A}}})
            +
    \boldsymbol{\theta}_{\widetilde{\mathcal{A}}}^{H}\widetilde{\mathbf{K}}\boldsymbol{\theta}
_{\widetilde{\mathcal{A}}}   + {\rm{const}},
            \end{aligned}
            \end{footnotesize}
    \end{align}
where {\rm{const}} is the constant irrelevant to $\boldsymbol{\theta}_{\widetilde{\mathcal{A}}}$. The definitions of $\widetilde{\mathbf{J}}$ and $\widetilde{\mathbf{K}}$ are provided in {\textbf{Lemma 3}}, and the other relevant constants are defined as $\mathbf{C}_k \triangleq {\rm{diag}}(\overline{\mathbf{h}}_k^{H}) \mathbf{a}_N \mathbf{a}_N^{H} {\rm{diag}}(\overline{\mathbf{h}}_k)$, $\mathbf{D}_{k,i} \triangleq {\rm{diag}}(\overline{\mathbf{h}}_k^{H})\mathbf{a}_N \mathbf{a}_N^{H}{\rm{diag}}(\overline{\mathbf{h}}_i)$, $\mathbf{F}_{k,i} \triangleq{\rm{diag}}(\overline{\mathbf{h}}_k^{H}){\rm{diag}}(\overline{\mathbf{h}}_i)$, and $m_{k,i} \triangleq\overline{\mathbf{h}}_k^{H} \overline{\mathbf{h}}_i$. Also, $s_{k,i}$, $n_{k,i}$, $l_{k,i}$, and $u_{k,i}$ are defined by extracting corresponding terms and omitted here for brevity. 
To derive the gradient w.r.t. $\boldsymbol{\theta}_{\widetilde{\mathcal{A}}}$, we first resort to the differentiation of $f_q$ w.r.t. $\boldsymbol{\theta}_{\widetilde{\mathcal{A}}}^{*}$. Accordingly, we have $d \big[ 
            (\boldsymbol{\theta}_{\widetilde{\mathcal{A}}}^{*} \otimes \boldsymbol{\theta}_{\widetilde{\mathcal{A}}})^{H}
            \widetilde{\mathbf{J}}
            (\boldsymbol{\theta}_{\widetilde{\mathcal{A}}}^{*} \otimes \boldsymbol{\theta}_{\widetilde{\mathcal{A}}})
        \big]
        =
            \big[
                (\boldsymbol{\theta}_{\widetilde{\mathcal{A}}}^{*}\otimes\boldsymbol{\theta}_{\widetilde{\mathcal{A}}})^{T}\widetilde{\mathbf{J}}^{T}(\mathbf{I}\otimes\mathbf{I})(\boldsymbol{\theta}_{\widetilde{\mathcal{A}}}\otimes\mathbf{I})
                +
                (\boldsymbol{\theta}_{\widetilde{\mathcal{A}}}^{*}\otimes\boldsymbol{\theta}_{\widetilde{\mathcal{A}}})^{H}\widetilde{\mathbf{J}}(\mathbf{I}\otimes\mathbf{I})(\mathbf{I}\otimes\boldsymbol{\theta}_{\widetilde{\mathcal{A}}})
            \big]^{T} d \boldsymbol{\theta}_{\widetilde{\mathcal{A}}}^{*}$, and
    $d \boldsymbol{\theta}_{\widetilde{\mathcal{A}}}^{H}\widetilde{\mathbf{K}}\boldsymbol{\theta}_{\widetilde{\mathcal{A}}}
        =
        \big[ 
            \boldsymbol{\theta}_{\widetilde{\mathcal{A}}}^{T}\widetilde{\mathbf{K}}^{T}
        \big]^{T} d \boldsymbol{\theta}_{\widetilde{\mathcal{A}}}^{*}$. Finally, by using 
    $\nabla f(\mathbf{z}) = 2 \frac{df(\mathbf{z})}{d\mathbf{z}^{*}}$ \cite{gradient}, we obtain the gradient w.r.t. $\boldsymbol{\theta}_{\widetilde{\mathcal{A}}}$ in \ref{RGA_gradient} which completes the proof.

\vspace{-10pt}  
\section{Proof of \textbf{Proposition 3}}

        Following the updating rule in \eqref{BCD_rule}, we have
        \begin{align} \label{non_increasing1}
        \begin{footnotesize}
        \begin{aligned}
                f_q(\boldsymbol{\theta}_{t}, \mathbf{p}_{t}, \boldsymbol{\eta}_{t}, \boldsymbol{\chi}_{t}) 
            \leq 
                f_q(\boldsymbol{\theta}_{t}, \mathbf{p}_{t+1}, \boldsymbol{\eta}_{t+1}, \boldsymbol{\chi}_{t+1}).
        \end{aligned}
        \end{footnotesize}
        \end{align}
        This holds since the closed-form updates of  $\boldsymbol{\eta}_{t+1}$, $\boldsymbol{\chi}_{t+1}$ and $\mathbf{p}_{t+1}$, are optimal when other variables are fixed as derived in \eqref{Optimal_eta}, \eqref{Optimal_chi} and \eqref{Optimal_power}, respectively.
        For simplicity, denote the parameter set as $\mathbf{W}_{t+1} = \{\mathbf{p}_{t+1}, \boldsymbol{\eta}_{t+1}, \boldsymbol{\chi}_{t+1} \}$. During the optimization w.r.t. $\boldsymbol{\theta}_t$ using \textbf{Algorithm 2} with other variables fixed, denoting the maximum number of iterations as $r_{m}$, we have $\boldsymbol{\theta}_{t,0} = \boldsymbol{\theta}_{t}$ and $\boldsymbol{\theta}_{t,r_{m}-1} = \boldsymbol{\theta}_{t+1}$. Besides, the following relationship holds
        \begin{subequations}
            \begin{align}
            \footnotesize 
            g_0(\boldsymbol{\theta}|\mathbf{W}_{t})
            &\leq
            \footnotesize
            g_1(\boldsymbol{\theta}; \boldsymbol{\theta}_t| \mathbf{W}_{t}), \label{a}
            \\
            \footnotesize 
            g_0(\boldsymbol{\theta}_{t}|\mathbf{W}_{t})
            &=
            \footnotesize 
            g_1(\boldsymbol{\theta}_t; \boldsymbol{\theta}_t| \mathbf{W}_{t}), \label{b}
            \\
            \footnotesize 
            g_1(\boldsymbol{\theta}; \boldsymbol{\theta}_{t}| \mathbf{W}_{t})
            &\leq
            \footnotesize 
            g_2(\boldsymbol{\theta}; \boldsymbol{\theta}_{t,r}| \mathbf{W}_{t}), \label{c}
            \\
            \footnotesize 
            g_1(\boldsymbol{\theta}_{t,r}; \boldsymbol{\theta}_{t}| \mathbf{W}_{t})
            &=
            \footnotesize 
            g_2(\boldsymbol{\theta}_{t,r}; \boldsymbol{\theta}_{t,r}| \mathbf{W}_{t}). \label{d}
            \end{align}
        \end{subequations}

        Accordingly, we have $-f_q(\boldsymbol{\theta}_{t}, \mathbf{p}_{t+1}, \boldsymbol{\eta}_{t+1}, \boldsymbol{\chi}_{t+1})
            =
            g_0(\boldsymbol{\theta}_{t}| \mathbf{W}_{t+1})
        \overset{a}{=}
            g_1(\boldsymbol{\theta}_{t}; \boldsymbol{\theta}_{t}| \mathbf{W}_{t+1})
            \overset{b}{=}
            g_2(\boldsymbol{\theta}_{t,0}; \boldsymbol{\theta}_{t,0}| \mathbf{W}_{t+1})
        \overset{c}{\geq}
            g_2(\boldsymbol{\theta}_{t+1}; \boldsymbol{\theta}_{t,r_m-2}| \mathbf{W}_{t+1})
            \overset{d}{\geq}
            g_1(\boldsymbol{\theta}_{t+1}; \boldsymbol{\theta}_{t}| \mathbf{W}_{t+1})
        \overset{e}{\geq}
            g_0(\boldsymbol{\theta}_{t+1}| \mathbf{W}_{t+1})
            =
            -f_q(\boldsymbol{\theta}_{t+1}, \mathbf{p}_{t+1}, \boldsymbol{\eta}_{t+1}, \boldsymbol{\chi}_{t+1}),
        $ \noindent where $(a)$, $(b)$, $(d)$, $(e)$ hold due to \eqref{b}, \eqref{d}, \eqref{c}, \eqref{a}, respectively. $(c)$ holds due to the update rule in \eqref{optimal_theta_MM}. Therefore, with the fact in \eqref{non_increasing1}, the objective value of $f_q(\boldsymbol{\theta}, \mathbf{p}, \boldsymbol{\eta}, \boldsymbol{\chi})$ is monotonically non-decreasing in the proposed \textbf{Algorithm 3} when \textbf{Algorithm 2} is applied for phase shift design. In addition, it can be inferred that $f_q(\boldsymbol{\theta}, \mathbf{p}, \boldsymbol{\eta}, \boldsymbol{\chi})$ is bounded due to power constraints, i.e., $p_i \leq p_{\max}, \forall i \in \mathcal{K}$. Thus, the proposed \textbf{Algorithm 3} with \textbf{Algorithm 2} for phase shift design finally converges to a stationary solution to problem $\mathcal{P}_1$ \cite{BCD_convergence}. For \textbf{Algorithm 1}, the monotonicity is guaranteed by the backtracking search method. The convergence is established following a similar procedure as mentioned above and is omitted here briefly. Thus, we complete the proof for \textbf{Proposition 3}.


\end{appendices}

\bibliographystyle{IEEEtran}
\bibliography{main}

\end{document}